\documentclass[a4paper,twocolumn,11pt,accepted=2022-05-02]{quantumarticle}
\pdfoutput=1
\usepackage[english]{babel}
\usepackage{latexsym}
\usepackage{graphics}
\usepackage{graphicx}
\usepackage{epsfig}
\usepackage{color}
\usepackage{bm}
\usepackage{amsmath}
\usepackage{amssymb}
\usepackage{amsthm}
\usepackage{dcolumn}
\usepackage{bm}
\usepackage{float}
\usepackage{hyperref}
\usepackage{color}
\usepackage{epstopdf}
\usepackage{cleveref}
\usepackage[svgnames]{xcolor}
\usepackage{enumerate}
\usepackage{braket}
\hypersetup{hidelinks,colorlinks=true,allcolors=DarkBlue}

\DeclareMathOperator{\Tr}{Tr}

\newtheorem{theorem}{Theorem}
\newtheorem{lemma}{Lemma}
\newtheorem{proposition}{Proposition}

\theoremstyle{remark}
\newtheorem{remark}{Remark}

\begin{document}

\title{Security of quantum key distribution with detection-efficiency mismatch in the multiphoton case}

\author{Anton Trushechkin}
\email{trushechkin@mi-ras.ru}
\homepage{http://www.mathnet.ru/eng/person31114}
\affiliation{Steklov Mathematical Institute of RAS, Steklov International Mathematical Center, Moscow 119991, Russia}
\affiliation{Department of Mathematics and NTI Center for Quantum Communications,\\
National University of Science and Technology MISIS, Moscow 119049, Russia}
\affiliation{QRate, Skolkovo, Moscow 143025, Russia}
\orcid{0000-0002-9541-0752}
\maketitle

\begin{abstract}
Detection-efficiency mismatch is a common problem in practical quantum key distribution (QKD) systems. Current security proofs of QKD with detection-efficiency mismatch rely either on the assumption of the single-photon light source on the sender side or on the assumption of the single-photon input of the receiver side. These assumptions impose restrictions on the class of possible eavesdropping strategies. Here we present a rigorous security proof without these assumptions and, thus, solve this important problem and prove the security of QKD with detection-efficiency mismatch against general attacks (in the asymptotic regime). In particular, we adapt the decoy state method to the case of detection-efficiency mismatch. 
\end{abstract}
\maketitle

\section{Introduction}

Quantum key distribution (QKD) is aimed to provide unconditionally secure communication. The notion of unconditional security means that an eavesdropper is allowed to have unlimited computational power. In theory, the eavesdropper is restricted only by quantum mechanics. The first QKD protocol was discovered by Bennett and Brassard in 1984 \cite{BB84} and is thus called the BB84 protocol. A number of unconditional security proofs for the BB84 protocol were proposed  \cite{Mayers0,Mayers,ShorPreskill, Renner,Koashi,GisinRenner,Toma}, which gave rise to a beautiful mathematical theory. 

Security proofs for practical implementations of QKD protocols faces with various problems caused by certain imperfections in apparatus setups \cite{Gisin,Scarani,nature2016,Xu}, which should be taken into account in security proofs. This paper is devoted to the problem of the efficiency mismatch between different threshold single-photon detectors.

In the BB84 protocol (as well as in other discrete-variable QKD protocols), information is encoded in the polarizations or phases of weak coherent pulses simulating true single-photon states. Single-photon detectors are used on the receiver side to read the information encoded in such states. Ideally, a single-photon detector fires whenever it is hit by at least one photon. However, a realistic detector is triggered by one photon only with a certain probability  $0<\eta\leq1$, which is referred to as the efficiency of a detector. Typical value of $\eta$ for the detectors used in practical QKD systems (based on avalanche photodiodes) is 0.1. The detectors based on superconductors have $\eta\approx0.9$, but they are more expensive and require cryogenic temperatures \cite{Gisin,Scarani}.

In this paper, we will consider the BB84 protocol with the active basis choice. In this case, two single-photon detectors are used on the receiver side: One for the signals encoding bit 0 and one for the signals encoding bit 1, respectively. Imperfect efficiencies of the detectors, i.e., $\eta<1$, is not a problem whenever the efficiencies of the detectors coincide with each other, because the loss in the detection rate can be treated as a part of transmission loss followed by ideal detectors with perfect efficiency. But even a small difference between the detectors' efficiencies make the aforementioned security proofs inapplicable. Since the detection loss is different for different detectors, we cannot anymore treat the detection loss as a part of the common transmission loss. This is a serious problem since, in practice, it is hard to build two detectors with exactly the same efficiencies \cite{Leuchs}. In this work, we assume that the detectors' efficiencies are constant (not fluctuating and not affected by the eavesdropper) and known to both legitimate parties and the eavesdropper. This means that we consider detection-efficiency mismatch due to manufacturing and setup, not by the eavesdropper's manipulations. A security proof for this case enhances the security of QKD based on the widely used BB84 protocol.

Detection-efficiency mismatch leads to unequal frequencies of zeros and ones in the so called raw key. Indeed, if, for example, the efficiency of the detector for the signals encoding bit 0 is higher than that of the detector for the signals encoding bit 1, then the frequency of zeros will be greater then the frequency of ones. This gives \textit{a priori} information on the raw key to the eavesdropper. Analysis of possible attacks on a QKD protocol with such \textit{a priori} information is challenging.

The first security proof for the BB84 protocol with detection-efficiency mismatch was proposed in Ref.~\cite{Fung} and generalized in Ref.~\cite{LyndSkaar}. In Ref.~\cite{Lutk-num}, higher secret key rates for BB84 with detection-efficiency mismatch were obtained numerically, which were confirmed analytically in Refs.~\cite{Bochkov,MaZhou}. Also the decoy state method was adopted to the case of detection-efficiency mismatch in Ref.~\cite{Bochkov}. In all these works, an essential restriction on the class of attacks allowed to the eavesdropper is imposed: she is not allowed to add photons to the single-photon pulses. However, in the general case, the technologically unlimited eavesdropper can increase the number of photons in the pulses, for example, to increase the probability of detection of the desired positions on the receiver side.

In the case of equal detection efficiencies, the reduction of the multiphoton (infinite-dimensional) input on the receiver side to the single-photon (finite-dimensional) input is provided by the so called squashing model \cite{squash0,squash1,squash}. Squashing model is a special map from an infinite-dimensional Hilbert space to a finite-dimensional one, which gives the same statistics of measurement results. However, the squashing model  for the detection-efficiency mismatch case, which was proposed only recently in Ref.~\cite{Zhang}, additionally requires a constraint that limits the fraction of the state with the number of photons exceeding a certain number. Hence, the multiphoton input on the receiver side should be analyzed explicitly.

In the mentioned recently published paper~\cite{Zhang}, the security of the BB84 protocol with detection-efficiency mismatch was analyzed for the multiphoton input on the receiver side (but the single-photon light source) using the numerical optimization techniques. The security is proved there under a conjecture, which is confirmed only numerically. Namely, in Ref.~\cite{Zhang}, it is conjectured that the minimal rate of double clicks  for $n>3$ photons on the receiver side is not smaller than this quantity for $n=3$. A similar conjecture for the minimal mean error rate is also used. Here we rigorously prove a variant of such conjecture for the double click rates using the entropic uncertainty relations, see Proposition~\ref{Propdb}. So, our Proposition~\ref{Propdb} may serve as the lacking part for the analysis in Ref.~\cite{Zhang}.

Another improvement of our work is an adaptation of the decoy state method to the case of detection-efficiency mismatch. In most implementations, the source (located at the lab of the sender) emits not true single-photon states but in weak coherent pulses, which make QKD vulnerable to the photon number splitting attack \cite{PNS,PNS2}. This problem can be fixed by the decoy state method, which  effectively allows us to bound the number of multiphoton pulses from above \cite{LoMa2005,Wang2005,MaLo2005,Ma2017,Trushechkin2016}. The usual decoy state method is formulated for the case of no detection-efficiency mismatch. So, a generalization of a security proof for the BB84 protocol with  detection-efficiency mismatch to the case of weak coherent pulses is not straightforward. In Ref.~\cite{Bochkov}, an adaptation of the decoy state method to the case of detection-efficiency mismatch was proposed under an additional assumption that the eavesdropper cannot add  photons to single-photon pulses. Here we relax this assumption and present an adaptation of the decoy state method for the case of detection-efficiency mismatch in the general case. 

Thus, we solve an important problem for practical QKD and rigorously prove the security of the BB84 protocol with detection-efficiency mismatch for the multiphoton Bob's input as well as adapt the decoy state method for this case (thus allowing the multiphoton Alice's output as well). The developed new methods are universal and can be used in other security proofs.

Finally, let us note that, even in the case of a constant detection-efficiency mismatch, the simple random discarding of some detection events of the detector with a higher efficiency, which is sometimes used \cite{Fung,AgnesiVilloresi}, is not a solution with proven security if the eavesdropper (Eve) can add photons to the single-photon pulses. Of course, this allows us to balance the number of ones and zeros in the raw key. However, this solution faces two problems. The first one is that it requires the precise knowledge of both efficiencies. In contrast, our method requires only the knowledge of a lower bound of the ratio $\eta$ of the lower efficiency to the higher one. Let us consider the second (and the main) problem, which concerns security. If Eve can add more photos to pulses, then she can again violate the balance between the number of ones and zeros by adding photons to the pulses and, thus, increase the probability of detection in the desired detector. Of course, by doing this, she introduces additional errors (which is clear intuitively and confirmed numerically in Refs.~\cite{Zhang,EntanglVerif}). So, it is not clear whether this strategy is advantageous for Eve. It depends on how much additional information per registered bit can Eve gain at the cost of additional errors. Nevertheless, there is no formal security proof for the case when such strategy is possible. Our proof takes into account this possible strategy by explicit consideration of a multiphoton (infinite-dimensional) Bob's input and bounding the multiphoton detection events, which are pessimistically assumed to be insecure.

Here we focus on the asymtotic regime (infinitely long keys). Theoretically, generalization to the finite-key case is straightforward using the quantum version of the de Finetti representation \cite{Renner}. But practically useful corrections are to be elaborated using, e.g., the entropy accumulation technique \cite{EntrAccum,EntrAccum2,EntrAccum3}.

The following text is organized as follows. In Sec.~\ref{SecPre}, we present preliminary information: a detection model and a brief description of the BB84 protocol. In Sec.~\ref{SecMain}, we formulate the results for the case of single-photon Alice's output and formulate the main theorem. In Sec.~\ref{SecProof}, we prove it. Finally, in Sec.~\ref{SecDecoy}, we  present an adaptation of the decoy state method.

\section{Preliminaries}\label{SecPre}

\subsection{Detection model}\label{SecDet}

We start with the description of a detection model. We will adopt the common names for the communication parties: Alice for the transmitting side, Bob for the receiver side, and Eve for the eavesdropper. Alice and Bob are also referred to as the legitimate parties. 

The information is encoded in quantum states of the two-dimensional single-photon Hilbert space $\mathbb C^2$. The elements of the standard basis ($z$ basis) will be denoted as $\ket0$ and $\ket1$. The elements of the Hadamard basis ($x$ basis) will be denoted as
\begin{equation}
\begin{split}
&\ket+=\frac{\ket0+\ket1}{\sqrt2}=H\ket0,\\
&\ket-=\frac{\ket0-\ket1}{\sqrt2}=H\ket1,
\end{split}
\end{equation} 
where $H$ is the so called Hadamard transformation.

The $n$-photon space is $(\mathbb C^2)^{\otimes_s n}$, where $\otimes_s$ is a symmetrized tensor product. If $n=0$, then $(\mathbb C^2)^{\otimes_s 0}=\mathbb C$ -- a one-dimensional complex vector space spanned by the vacuum vector $\ket{\rm vac}$. An arbitrary number of photons corresponds to the bosonic Fock space $\mathcal F(\mathbb C^2)=\bigoplus_{n=0}^\infty (\mathbb C^2)^{\otimes_s n}$. Denote $c^\dag_{z0}$ ($c_{z0}$), $c^\dag_{z1}$ ($c_{z1}$), $c^\dag_{x0}$ ($c_{x0}$), and $c^\dag_{x1}$ ($c_{x_1}$) the creation (annihilation) operators of a photon in the states $\ket0$, $\ket1$, $\ket+$, and $\ket-$, respectively.  They are related to each other as 
\begin{equation}\label{Eqcxz}
c_{xa}=\frac{c_{z0}+(-1)^ac_{z1}}{\sqrt2},\quad a=0,1.
\end{equation}
Denote 
\begin{equation}\label{Eqketn}
\ket{n_0,n_1}_b=
\frac{(c_{b0}^\dag)^{n_0}}{\sqrt{n_0!}}
\frac{(c_{b1}^\dag)^{n_1}}{\sqrt{n_1!}}
\ket{\rm vac}.
\end{equation}
The $\mathbb C^2$ space is naturally embedded into $\mathcal F(\mathbb C^2)$ if we identify $\ket a=c^\dag_{za}\ket{\rm vac}$ and $H\ket a=c^\dag_{xa}\ket{\rm vac}$, $a\in\{0,1\}$.

Consider the case of the perfect detection, i.e., each detector fires whenever it is hit by at least one photon. Then, the measurement in the basis $b\in\{z,x\}$ is described by the positive operator-valued measure (POVM; in book \cite{Holevo} it is argued that a more precise term for this notion is ``probability operator-valued measure'') $\{\tilde P_\varnothing,\tilde P_0,\tilde P_1,\tilde P_{01}\}$, where the operators correspond to four possible outcomes: no click, only detector 0 clicks, only detector 1 clicks, and both detectors click (double click), and
\begin{equation}\label{EqPOVM}
\begin{split}
&\tilde P_\varnothing^{(b)}=\ket{\rm vac}\bra{\rm vac},\\
&\tilde P^{(b)}_0=\sum_{n_0=1}^\infty\ket{n_0,0}_b\bra{n_0,0},\\
&\tilde P^{(b)}_1=\sum_{n_1=1}^\infty\ket{0,n_1}_b\bra{0,n_1},\\
&\tilde P^{(b)}_{01}=\sum_{n_0,n_1=1}^\infty\ket{n_0,n_1}_b\bra{n_0,n_1}.\\
\end{split}
\end{equation}

Now let us describe the imperfect detection. Let $\eta_0$ and $\eta_1$ be the efficiencies of the detectors and, say, $0<\eta_1\leq\eta_0\leq1$. If we adopt an approximation that the imperfect detection can be modeled by an asymmetric beam splitter followed by a perfect detection, then, as shown in Ref.~\cite{EntanglVerif}, the efficiencies can be renormalized as $\eta'_0=1$ and $\eta'_1=\eta_1/\eta_0=\eta$. The common loss $1-\eta_0$ in both detectors can be treated as  additional transmission loss. So, in the aforementioned approximation, without loss of generality, we assume that $\eta_0=1$ and $\eta_1=\eta$, $0<\eta\leq1$. Then the imperfect detection can be described by the POVM $\{P_\varnothing,P_0,P_1,P_{01}\}$ \cite{EntanglVerif}:

\begin{equation}\label{EqPOVMtild}
\begin{split}
&
P_\varnothing^{(b)}=
\sum_{n_1=0}^\infty
(1-\eta)^{n_1}\ket{0,n_1}_b\bra{0,n_1},\\
&P^{(b)}_0=
\sum_{n_0=1}^\infty\sum_{n_1=0}^\infty
(1-\eta)^{n_1}\ket{n_0,n_1}_b\bra{n_0,n_1},\\
&P^{(b)}_1=
\sum_{n_1=1}^\infty
[1-(1-\eta)^{n_1}]\ket{0,n_1}_b\bra{0,n_1},\\
&P^{(b)}_{01}=
\sum_{n_0,n_1=1}^\infty
[1-(1-\eta)^{n_1}]\ket{n_0,n_1}_b\bra{n_0,n_1}.\\
\end{split}
\end{equation}

The mismatch parameter $\eta$ is assumed to be constant and known to both legitimate parties and the eavesdropper. This means that we consider detection-efficiency mismatch due to manufacturing and setup, not by Eve's manipulations. Denote also the probability that detector~1 clicks when it is hit by exactly $n$ photons:
\begin{equation}\label{EqTheta}
\theta_n=1-(1-\eta)^{n}.
\end{equation}

In fact, as we will discuss later in Remark~\ref{RemNonLinDet}, our analysis will rely neither on the precise form (\ref{EqPOVMtild}) of the POVM nor on the precise formula (\ref{EqTheta}) for $\theta_n$. However, we will essentially use the reduction to the case $\eta_0=1$. So, we assume that detector 0 (the detector with a larger efficiency) can be modeled by a beam splitter with the transmission coefficient $\eta_0$ followed by a perfect detection, but detector 1 is not assumed to be equivalent to this model. 

Let us explain the last statement. Denote $\theta^{(0)}_n$ ($\theta^{(1)}_n$) the original probability that detector~0 (1) clicks whenever it is hit by exactly $n$ photons. Here, ``original'' means  before we factor out the common noise. We  assume that:
\begin{enumerate}[(i)]
\item $\theta^{(0)}_n=1-(1-\eta_0)^{n}$, where $\eta_0$ is the original efficiency of detector~0. That is, detector~0 is modeled by a beam splitter with the transmission coefficient $\eta_0$ followed by a perfect detection.

\item  $\theta^{(1)}_n\leq \theta^{(0)}_n$, i.e., for any number of arriving photons, the ``$n$-photon efficiency'' of detector~1 does not exceed the corresponding efficiency of detector~0.
\end{enumerate}
This means that we can virtually set a beam splitter with the transmission coefficient $\eta_0$  before the detection scheme and consider detector~0 perfect. However, there might be additional loss on detecor~1 not necessarily described by another beam splitter. Then we put $\theta_n=\theta_n^{(1)}/\theta_n^{(0)}$ and $\eta=\min_{n\geq1}\theta_n$. The dependence of $\theta_n$ on $n$ may be arbitrary. 

Such generalization will require only minor modifications of the derived formulas. Namely, in some cases we will use the explicit expression (\ref{EqTheta}) for $\theta_2$. However, we can substitute it by an upper or a lower bound since $\eta\leq\theta_2\leq1$.

\subsection{Brief description of the BB84 protocol}\label{SecBB84}

In this subsection, we briefly describe the BB84 protocol. For more details see, e.g., Refs.~\cite{Gisin,Scarani,Fung2018}. We start with the prepare\&measure formulation, where Alice prepares single-photon states and sends them to Bob, who measures them.

%\begin{enumerate}[(1)]
%\item 
(1)~~Alice randomly chooses $N$ bases $\mathbf b=b_1\ldots b_N$, $b_i\in\{z,x\}$. These random choices are independent and identically distributed. Denote $p_z$ and $p_x=1-p_z$ probabilities of choosing the $z$ basis and the $x$ basis, respectively. Bob also randomly generates $N$ bases $\mathbf b'=b'_1\ldots b'_N$, $b'_i\in\{z,x\}$ with the same distribution, independently of the Alice's choices.

%\item 
(2)~~Alice randomly generates a large number $N$ of bits $\mathbf a=a_1\ldots a_N$, $a_i\in\{0,1\}$. These random bits are independent and uniformly distributed. Then Alice prepares $N$ quantum states $\ket{a_i}\in\mathbb C^2$ for $b_i=z$ and $H\ket{a_i}$ for $b_i=x$, $i=1,\ldots,N$, and sends them to Bob. Bob measures them according to POVM (\ref{EqPOVMtild}) with $b=b_i$ for a given position $i$ and records the results $\mathbf a'=a'_1\ldots a'_N$: no click corresponds to $a'_i=\varnothing$, a click of  detector 0 or 1 corresponds to $a'_i=0$ or $a'_i=1$, respectively. In case of a double click, Bob randomly (with equal probabilities) chooses $a'_i=0$ or $a'_i=1$. But also he can count the number of double clicks  (this will be used in the analysis). The strings $\mathbf a$ and $\mathbf a'$ are referred to as the \textit{raw keys}.

%\item 
(3)~~Using a public authentic classical channel, Alice and Bob announce their bases (the strings $\mathbf b$ and $\mathbf b'$), also Bob announces the positions where he has obtained a click (no matter single or double). We adopt a version of the protocol where only the $z$ basis is used for key generation. So, Alice and Bob keep positions where $b_i=b'_i=z$ and Bob obtained a click. Denote the strings with only such positions as $\tilde{\mathbf a}$ and $\tilde{\mathbf a}'$, which are referred to as the \textit{sifted keys}. The other (dropped) positions still can be used for  estimation of the achievable key generation rate. 

%\item 
(4)~~Communicating over the public authentic classical channel, Alice and Bob estimate the Eve's information about the Alice's sifted key and either abort the protocol (if this information is too large) or perform procedures of error correction and privacy amplification to obtain a \textit{final key}. The former allows them to fix the discrepancies between their sifted keys. Here we adopt a version where positions with $\tilde a_i\neq\tilde a'_i$ are treated as the Bob's errors. During the error correction procedure, Bob corrects these errors and obtains the key identical to the Alice's one, i.e., $\tilde{\mathbf a}$.

In the privacy amplification procedure, Alice and Bob apply a randomly generated special map (a hash function)  to the sifted key and obtain a shorter key, which is a final key. Alice and Bob must have an upper bound for the Eve's information about the Alice's sifted key $\tilde{\mathbf a}$ (or, equivalently, a lower bound on the Eve's ignorance about the Alice's sifted key) to calculate the length $l$ of the final key such that the eavesdropper has only an infinitesimal (as $N\to\infty$) information about the final key. The corresponding ratio $l/N$ will be referred to as the \textit{secret key rate}. See Ref.~\cite{DW} for a formal definition of secret key rate. Here we do not need a formal definition. Actually, we will use formula (\ref{EqDW}) below as the starting formula for the secret key rate. Note that sometimes the secret key rate is defined as the ratio of $l$ to the length of the sifted key (rather than the raw one).
%\end{enumerate}

In this paper, we consider only the asymptotic case $N\to\infty$ and do not address the finite-key effects. In this case, we can put $p_x\sim1/\sqrt{N}\to0$, $p_z\to1$.

For mathematical analysis of the security of the protocol, it is convenient to reformulate it in terms of an equivalent entanglement-based version. In the entanglement-based version of the protocol, step (2) in the description above is altered. Alice does not generate the string $\mathbf a$ and does not prepare and send quantum states. Instead, a source of entangled states generates a state
\begin{equation}\label{EqPhi}
\rho_{AB}=\ket{\Phi^+}\bra{\Phi^+}
\end{equation}
in the Hilbert space $\mathbb C^2\otimes\mathcal F(\mathbb C^2)$,
where 
\begin{equation*}
\begin{split}
\ket{\Phi^+}&=\frac1{\sqrt2}(\ket{0}\otimes\ket{1,0}_z+\ket{1}\otimes\ket{0,1}_z)\\
&=\frac1{\sqrt2}(\ket{+}\otimes\ket{1,0}_x+\ket{-}\otimes\ket{0,1}_x)
\\
&\in\mathbb C^2\otimes\mathcal F(\mathbb C^2).
\end{split}
\end{equation*}
and sends the first subsystem (a qubit) to Alice and the second one to Bob. Then, like Bob, Alice performs a measurement in the basis $b_i$ (the POVM $\{\ket0\bra0,\ket1\bra1\}$ for the $z$ basis and $\{\ket+\bra+,\ket-\bra-\}$ for the $x$ basis) and records the result $a_i$. Since Alice's measurement is virtual, her measurement corresponds to detectors with the perfect efficiencies. This scheme is also referred to as the source replacement scheme \cite{BBM,CLL,FL}.

Let us highlight that the notations $\ket0$ and $\ket1$ denote logical bits, while notations for optical modes are of the form $\ket{n_0,n_1}_b$, i.e., firstly, always include two modes and, secondly, include the subindex $b$ denoting the basis.

In the entanglement-based QKD, Eve is assumed to control the source of entangled states, i.e., she can replace the density operator $\rho_{AB}$ given in (\ref{EqPhi}) by her own arbitrary density operator $\rho_{ABE}$ acting on the Hilbert space $\mathbb C^2\otimes\mathcal F(\mathbb C^2)\otimes\mathcal H_E\equiv \mathcal H_A\otimes\mathcal H_B\otimes\mathcal H_E$, where $\mathcal H_E$ is an arbitrary separable Hilbert space. Eve is assumed to own the additional register $E$ corresponding to the space $\mathcal H_E$. The transmission loss in the original prepare\&measure formulation can be included in the state $\rho_{ABE}$. But the detection loss is basis-dependent in the case of  detection-efficiency mismatch and cannot be included in $\rho_{ABE}$. 

Let us illustrate the last point. Let, e.g., $\rho_B=\ket0\bra0$. If Bob chooses the $z$ basis, then this state goes to the ideal detector~0 and the detection probability is equal to one. However, if Bob chooses the $x$ basis, then, with equal probabilities, this state goes to either detector~0 (detection with probability~1) or detector~1 (detection with probability $\eta$). So, the mean detection probability is $(1+\eta)/2$. So, the mean detection probability depends on the basis and, hence, cannot be incorporated into the state $\rho_{ABE}$ generated before the choices of the bases.

\section{Formulation of results for the case of single-photon Alice's output}\label{SecMain}

\subsection{Problem statement}\label{SecProblem}

Let us adopt the following common agreement: For any tripartite density operator $\rho_{ABE}$, the notations like $\rho_{AB}$, $\rho_B$, etc. mean $\rho_{AB}=\Tr_E\rho_{ABE}$, $\rho_B=\Tr_{AE}\rho_{ABE}$, etc. Denote $\mathfrak T(\mathcal H)$ the space of trace-class operators on a Hilbert space $\mathcal H$. If $\mathcal H$ is finite-dimensional, then $\mathfrak T(\mathcal H)$ coincides with the space of all linear operators on $\mathcal H$. If $\Phi$ is a quantum transformation acting on, for example, Bob's Hilbert space, i.e., $\Phi$ is a linear map from $\mathfrak T(\mathcal F(\mathbb C^2))$ to itself, then  $\Phi(\rho_{ABE})\equiv({\rm Id}_A\otimes\Phi\otimes{\rm Id}_E)(\rho_{ABE})$ with ${\rm Id}$ being the identity quantum transformation on the corresponding Hilbert space, i.e., the identity operators in the spaces $\mathfrak T(\mathcal H_{A,B,E})$. Also denote ${\rm Id}_{BE}={\rm Id}_B\otimes{\rm Id}_E$. 

Let, as before, $\rho_{ABE}$ be a tripartite density operator corresponding to a sending.  Denote
\begin{equation}\label{EqRhoPrime}
\rho'_{ABE}=\mathcal G(\rho_{ABE})\equiv G\rho_{ABE}G,
\end{equation}
where
\begin{equation}\label{EqG}
G=\sqrt{I_B-P_\varnothing^{(z)}}.
\end{equation}
$I_{A,B,E}$ are the identity operators in the corresponding Hilbert spaces (not to be confused with ${\rm Id}_{A,B,E}$, which denote the identity ``superoperators'': the identity operators in the spaces of the trace-class operators on Alice's, Bob's and Eve's Hilbert spaces). Also, we will use the denotation $I_{BE}=I_B\otimes I_E$. Transformation~(\ref{EqRhoPrime}) corresponds to the following partial measurement: Instead of the full measurement in the $z$ basis described by Eq.~(\ref{EqPOVMtild}), we just check whether we obtain a detection or not. If we do not obtain a detection, we sift this position out. In other words, Eq.~(\ref{EqRhoPrime}) is a post-selection map. Then
\begin{equation*}
p_{\rm det}=\Tr\rho'_{ABE}\leq1
\end{equation*}
is the detection probability for Bob if he measures in the $z$ basis. Since Bob announces the positions where he has obtained a click, this quantity is known to the legitimate parties. Denote also $\tilde\rho'_{ABE}=p_{\rm det}^{-1}\rho'_{ABE}$ the corresponding normalized state.

The Alice's measurement in the $z$ basis and the $x$ basis can be described by the decoherence  maps $\mathcal Z$ and $\mathcal X$ in the corresponding bases:
\begin{equation}\label{EqZX}
\begin{split}
\mathcal Z(\rho_{A})&=\sum_{a=0}^1
\ket a\bra a
\rho_{A}
\ket a\bra a,\\
\mathcal X(\rho_{A})&=
\sum_{a=0}^1
(H\ket a\bra aH)
\rho_{A}
(H\ket a\bra aH).
\end{split}
\end{equation}
Denote
\begin{equation}
\begin{split}
\tilde\rho'_{ZBE}&=\mathcal Z(\tilde\rho'_{ABE}),\\
\tilde\rho'_{XBE}&=\mathcal X(\tilde\rho'_{ABE}).
\end{split}
\end{equation}

According to the Devetak--Winter theorem \cite{DW}, the secret key rate is given by
\begin{equation}\label{EqDW}
K=p_{\rm det}[H(Z|E)_{\tilde\rho'}-H(Z|B)_{\tilde\rho'}],
\end{equation}
where the conditional von Neumann entropies are calculated for the state $\tilde\rho'_{ZBE}$. Here we have taken into account that only the positions where both legitimate parties used $z$ basis participate in the sifted key and $p_z\to1$ as $N\to\infty$.

Let us comment the applicability of the Devetak--Winter theorem to the case of detection-efficiency mismatch. This is a general theorem and is formulated in terms of an abstract tripartite state. In the state $\tilde\rho'_{ZBE}$, the detector loss has been already taken into account. In fact, the introduction of the map $\mathcal G$ and the state  $\tilde\rho'_{ZBE}$ means that we have represented an imperfect detection as a detection loss followed by a perfect detection. The state $\tilde\rho'_{ZBE}$ is the state after the detection loss and after sifting the unregistered positions, but before the perfect detection. Hence, the Devetak--Winter theorem can be applied to this tripartite state. The direct application of this theorem gives formula (\ref{EqDW}) without the prefactor $p_{\rm det}$, but such expression would correspond to the definition of the secret key rate as the ratio of the final key length to the sifted key length (since the state $\tilde\rho'_{ZBE}$ corresponds to the sifted key). Since here we define the secret key rate as a ratio of the final key length to the number of the emitted pulses, we should multiply the Devetak--Winter formula by the ratio of the sifted key length to the number of the emitted pulses, which is exactly the detection rate $p_{\rm det}$.

The first and the second terms in the brackets in Eq.~(\ref{EqDW}) characterize the Eve's and  Bob's ignorances about the Alice's sifted key bit, respectively. Here we assume that the length of the error-correcting syndrome is given by the Shannon theoretical limit. Otherwise, a factor $f>1$ should be added to the second term. The present-day error-correcting codes allow for $f=1.22$. An error-correction procedure for QKD based on the low-density parity-check codes, which decreases the factor $f$, is given in Refs.~\cite{SymBl1,SymBl2}. A syndrome-based QBER estimation algorithm, which also can decrease $f$, is proposed in Ref.~\cite{fly}.

Using Fano's inequality, the second term in the right-hand side of Eq.~(\ref{EqDW}) is bounded from above by $h(Q_z)$, where $h(x)=-x\log x-(1-x)\log(1-x)$ ($\log\equiv\log_2$) is the binary entropy and $Q_z$ is the quantum bit error rate (QBER), i.e., error ratio in the $z$ basis. This value is observed by Alice and Bob, hence, thus term is known by the legitimate parties. Formally, $Q_z$ is defined as
\begin{equation}\label{EqQz}
\begin{split}
Q_z&=\Tr\tilde\rho'_{AB}(\ket0\bra0\otimes P_1^{(z)}
\\&+\ket1\bra1\otimes P_0^{(z)}+I_A\otimes P_{01}^z/2).
\end{split}
\end{equation}
The last term here means that, in the case of a double click, a bit value 0 or 1 is assigned with the probability 1/2. So, 1/2 is the error probability in the case of a double click.

But the first term (the Eve's ignorance) in the right-hand side of Eq.~(\ref{EqDW}) is unknown to Alice and Bob since the state $\rho_{ABE}$ is chosen by Eve. They should estimate the Eve's ignorance from below using the observable data. To simplify the problem and eliminate the dependence of the right-hand side of Eq.~(\ref{EqDW}) on the Eve's subsystem, we apply the entropic uncertainty relations \cite{BertaEUR,ColesEUR}:
\begin{equation}
H(Z|E)_{\tilde\rho'}+H(X|B)_{\tilde\rho'}\geq1,
\end{equation}
where the second conditional entropy is calculated for the state $\tilde\rho'_{XBE}$, or, equivalently, $\tilde\rho'_{XB}$.

So, the secret key rate is lower bounded by 
\begin{equation}\label{EqDWmin}
K\geq p_{\rm det}[1-\sup_{\rho_{AB}\in\mathbf S}H(X|B)_{\tilde\rho'}-h(Q_z)],
\end{equation}
where 
\begin{multline}\label{EqS}
\mathbf S=\{\rho_{AB}\in\mathfrak T(\mathcal H_A\otimes\mathcal H_B)\|
\,\rho\geq0,\:\Tr\Gamma_i\rho_{AB}=\gamma_i,\\i=1,\ldots,m\}.
\end{multline}
Here $\Gamma_i$ are linear operators acting on $\mathcal H_A\otimes\mathcal H_B$ corresponding to observables of Alice and Bob. They impose constraints on $\rho_{AB}$ since the latter should be consistent with the observed data. The detection probability $p_{\rm det}$ is one of the observed quantities and, thus, fixed, see Eq.~(\ref{EqGamma1}) below.

\begin{remark}\label{RemPhaseErr}
Note that, in Eq.~(\ref{EqDWmin}), we should estimate the Bob's ignorance about the Alice's outcome in the $x$ basis for the state $\tilde\rho'_{AB}$, i.e., after the attenuation corresponding to the measurement in the $z$ basis, see Eq.~(\ref{EqRhoPrime}). In the case of equal detector efficiencies, Alice and Bob can estimate this entropy in the same way as the last term $H(Z|B)$ in Eq.~(\ref{EqDW}). Namely, Bob just measures his state in the $x$ basis, calculates the ratio of discrepancies between his and the Alice's results in the same basis, and use the Fano's inequality for the estimation of $H(X|B)$. But in the case of  detection-efficiency mismatch, he cannot proceed in such a way. Indeed, now the detection loss is basis-dependent, and the imperfect measurement in the $x$ basis leads to the state $G^{(x)}\rho_{AB}G^{(x)}$, where $G^{(x)}=\sqrt{I_B-P_\varnothing^{(x)}}$ instead of $\rho'_{AB}=G\rho_{AB}G$, where $G$ is given by Eq.~(\ref{EqG}). This is the main difficulty in the security analysis for the case of detection-efficiency mismatch. The inconsistency of the phase error rate with the bit error rate in the $x$ basis is a common problem for QKD with device imperfections \cite{GLLP}.
\end{remark}

Also note that the conditional entropy can be expressed as
\begin{equation}\label{EqHD}
H(X|B)_{\tilde\rho'}=-D(\tilde\rho'_{XB}\|I_A\otimes\tilde\rho'_B),
\end{equation}
where $D(\sigma\|\tau)=\Tr\sigma\log\sigma-\Tr\sigma\log\tau$ is the quantum relative entropy, which is a jointly convex function. So, problem (\ref{EqDWmin}) is a convex minimization problem subject to linear constraints. In the case of finite-dimensional $\mathcal H_A\otimes\mathcal H_B$, it can be solved numerically, as was proposed in Refs.~\cite{Lutk-num,Lutk-num-pre}. But we have an infinite-dimensional $\mathcal H_B$, so, an analytic bound for $H(X|B)$ (or, at least, an analytic reduction to a finite-dimensional optimization problem) is required for a rigorous security proof.

\subsection{Main theorem and simulation}\label{SecMainTh}

Consider the following linear constraints:
%\begin{itemize}
%\item 

(1)~~Probability of detection (for the $z$ basis)
\begin{equation}\label{EqGamma1}
\Gamma_1=I_A\otimes(I_B-P_\varnothing^{(z)}),
\quad \Tr\Gamma_1\rho_{AB}=p_{\rm det}.
\end{equation}

%\item 
(2)~~Weighted mean erroneous detection rate in the $x$ basis
\begin{equation}\label{EqGamma2}
\begin{split}
&\Gamma_2=
\eta^{-1}\ket+\bra+\otimes
\left(P^{(x)}_1+\frac12P^{(x)}_{01}\right)
\\
&\quad\,+
\ket-\bra-\otimes
\left(P^{(x)}_0+\frac12P^{(x)}_{01}\right),
\\
&\Tr\Gamma_2\rho_{AB}=q.
\end{split}
\end{equation}
This means that, in the case of a double click, a bit value 0 or 1 is assigned with the probability 1/2, but we consider a weighted sum: erroneous ones are taken with the weight $\eta^{-1}$ and erroneous zeros are taken with the weight 1. Also note that $q$ is not the usual QBER. This is a weighted sum of probabilities of erroneous detections rather than the ratio of erroneous detections to all detections (cf. Eq.~(\ref{EqQz}), where the state $\tilde\rho'_{AB}=p_{\rm det}^{-1}\rho'_{AB}$ instead of $\rho'_{AB}$ is used). 

A weighted sum can be informally explained as follows. The error rate in the $x$ basis is used to estimate Eve's potential information on the sifted key. However, for this aim, we need the error rate which would be in the case of perfect detection, which is not directly observable. Loss on detector~1 reduces the observable error rate. So, in order to estimate the ``actual'' error rate from above (a pessimistic bound), we should divide the observable error rate on detector~1 by the detector efficiency $\eta$. If all photons arriving at Bob's lab are single photons, such division exactly restores the ``actual'' error rate (i.e., which would be observed in the case of perfect detection). If some arriving pulses are multiphoton, which have higher probabilities of detection, such division gives an upper bound of the ``actual'' error rate.

%\item 
(3)~~Probability of a single click of  detector~1 for the measurement in the $z$ basis
\begin{equation}\label{EqGamma3}
\Gamma_3=
I_A\otimes P^{(z)}_{1},
\quad
\Tr\Gamma_3\rho_{AB}=p_1.
\end{equation}

%\item 
(4)~~Mean probability of a double click
\begin{equation}\label{EqGamma4}
\Gamma_4=
I_A\otimes\frac12(P^{(z)}_{01}+P^{(x)}_{01}),
\quad
\Tr\Gamma_4\rho_{AB}=p_{01}.
\end{equation}
%\end{itemize}

We will also use the quantity
\begin{equation}
t=\Tr\tilde\Gamma_1\rho_{AB},\quad
\tilde\Gamma_1=I_A\otimes(I_B-\tilde P_\varnothing^{(z)}),
\end{equation}
where, recall, $\tilde P_\varnothing^{(z)}$ was defined in Eq.~(\ref{EqPOVM}) as an operator from the POVM corresponding to the ideal detection. So, $t$ is the detection rate which would be observed in the case of ideal detection and, thus, is not directly observable. The notation $t$ comes from ``transparency'' since it actually denotes the transparency of the transmission line and does not take into account the detection loss.

Analogously to $p_1$, let us define $p_{0+01}$: the probability of either a single click of detector~0 or a double click, so that $p_{\rm det}=p_{0+01}+p_1$.

Further, since the number of photons arriving at Bob's side is arbitrary, the mentioned quantities $p_{\rm det}$, $q$, $p_1$, $t$, $p_{0+01}$ and  can be decomposed into sums of contributions from the $n$-photon pulses arriving at Bob's sides:
\begin{subequations}\label{EqDecompositions}
\begin{gather}
p_{\rm det}=\sum_{n=1}^\infty p_{\rm det}^{(n)},
\quad
q=\sum_{n=1}^\infty q_n,
\quad t=\sum_{n=1}^\infty t_n,
\\
p_1=\sum_{n=1}^\infty p_1^{(n)},
\quad
p_{0+01}=\sum_{n=1}^\infty p_{0+01}^{(n)}.
\end{gather}
\end{subequations}
The same is true for $p_{01}$, but we will not use such decomposition for this quantity. Also, we will use the notations 
\begin{equation*}
p_{\rm det}^{(3+)}=\sum_{n=3}^\infty p_{\rm det}^{(n)},
\quad
p_1^{(3+)}=\sum_{n=3}^\infty p_1^{(n)},
\quad
t_{3+}=\sum_{n=3}^\infty t_n.
\end{equation*}

Let us establish relations between $t_n$ and $p_{\rm det}^n$: Since $t$ is not directly observable, we will need certain estimations. We have
\begin{equation}\label{Eqpdetp}
p_{\rm det}^{(n)}=p_{0+01}^{(n)}+p_1^{(n)}
\end{equation}
and
\begin{equation}\label{Eqtp}
t_n=p_{0+01}^{(n)}+\frac{p_1^{(n)}}{\theta_n}.
\end{equation}
Let us explain the last relation. Imperfect detection of  detector~1 turns some potential double clicks into single clicks of  detector~0 and some potential single clicks of  detector~1 into no clicks. So, only the signals which lead to a single click of   detector~1 in the case of perfect detection can be lost in the case of the imperfect one. The terms $p_1^{(n)}/\theta_n$ mean that the fraction $1-\theta_n$ is lost on detector~1 from the $n$-photon part. 

In view of Eqs.~(\ref{Eqpdetp}) and~(\ref{Eqtp}), we have
\begin{equation}\label{Eqtnineq}
p^{(n)}_{\rm det}\leq t_n\leq \frac{p^{(n)}_{\rm det}}{\theta_n}
\end{equation}
and
\begin{equation}\label{Eqt1}
t_1=p_{\rm det}^{(1)}
+p_1^{(1)}\left(\frac1{\eta}-1\right).
\end{equation}

Our security proof has two parts: analytic bounds of Eve's ignorance for the single-photon part and analytic bounds for the fraction of multiphoton pulses. The second part is essentially based on the following estimates we are going to prove:

(I)~~If a pulse arriving at Bob's side contains three or more photons, the probability of a double click at least in one basis ($z$ or $x$) is strictly positive (this will be proved in Lemma~\ref{Lemdbmain}). This is true whenever efficiencies of both detectors are strictly positive. Moreover, in the case of perfect detection, the mean  probability of a double click is lower bounded by the unique root $p_{01}^{\min}$ of the equation (see Lemma~\ref{Lemdbmonot})
\begin{equation}\label{Eqp01min}
2p_{01}^{\rm min}\log3+2h(p_{01}^{\rm min})=1.
\end{equation}
Numerically, $p_{01}^{\rm min}\approx0.06$ ($p_{01}^{\rm min}>0.06$). Consequently, the probability of a double click in the case of imperfect detection is lower bounded by $\eta p_{01}^{\min}$ since at most the fraction $1-\eta$ of double clicks may not occur due to the detection loss. Hence,
\begin{equation}
t_{3+}\leq\frac{p_{01}}{\eta p_{01}^{\min}}
\end{equation}
(see Proposition~\ref{Propdb} for details) and, in view of Ineq.~(\ref{Eqtnineq}),
\begin{equation}\label{Eqpdet3u}
p_{\rm det}^{(3+)}\leq t_{3+}\leq \frac{p_{01}}{\eta p_{01}^{\rm min}}
=
p_{\rm det}^{(3+),\rm U}.
\end{equation}

(II)~~If a pulse contains only two photons, then the mean  probability of a double click may be zero. Namely, the two-photon Bell state
\begin{eqnarray}
\ket{\Phi^+}&=&\frac1{\sqrt2}(\ket{2,0}_z+\ket{0,2}_z)\nonumber\\
&=&\frac1{\sqrt2}(\ket{2,0}_x+\ket{0,2}_x)\label{EqBellPhi}
\end{eqnarray}
produces no double clicks in both bases. So, estimation of the fraction of two-photon pulses requires a separate analysis. 

The intuition behind the analysis is as follows. If the number of double clicks is small, then the two-photon part of the Bob's state is close to the pure state $\ket{\Phi^+}$. So,  Bob's subsystem is almost uncorrelated with  Alice's one and the QBER is close to 1/2. So, the double click rate impose restrictions to $q_2$ and also to $p_1^{(2)}$. The following estimates are satisfied (based mainly on Proposition~\ref{PropQ2}):

\begin{equation}\label{Eqq2L}
q_2\geq q_2^{\rm L}=
\max\left[
\frac{1+\theta_2/\eta}4p_{\rm det}^{(2)}-
\sqrt{\frac{2\theta_2p_{01}p_{\rm det}^{(2)}}{\eta^3}},\,
0\right]
\end{equation}
and
\begin{equation}\label{Eqp12U}
p_1^{(2)}\leq p_1^{(2),\rm U}=
\min\left[
\frac{p_{\rm det}^{(2)}}2+
\sqrt{\frac{2\theta_2p_{01}p_{\rm det}^{(2)}}{\eta}},
\,
p_{\rm det}^{(2)}
\right].
\end{equation}

The following bounds follow from Eqs.~(\ref{EqDecompositions}) and basic bounds (\ref{Eqpdet3u}), (\ref{Eqq2L}), and (\ref{Eqp12U}). Firstly, we will need a lower bound for $p_{\rm det}^{(1)}$ since it determines the number of positions in the sifted key treated as ``secure''. That is, Eve cannot obtain information about  such positions without introducing errors, while all multiphoton positions are treated as insecure. It turns out that
\begin{equation}\label{Eqpdet1l}
p_{\rm det}^{(1)}\geq p_{\rm det}^{(1),\rm L}=
p_{\rm det}-p_{\rm det}^{(2)}-p_{\rm det}^{(3+),\rm U},
\end{equation}
Secondly, Eve's knowledge on the single-photon positions can be estimated from below using the upper bound on the fraction of errors $q_1/t_1$. This is a modification of the standard intuition behind the BB84 protocol: Eavesdropping in the $z$ basis  leads to errors in the $x$ basis. Hence, we should estimate the numerator from above and the denominator from below. Formally, this can be understood from Theorem~\ref{ThMain} and Eqs.~(\ref{EqMain}) and~(\ref{EqDeltaxl}) below. Strictly speaking, they do not contain the fraction $q_1/t_1$, but, nevertheless, we can see that the relation between $q_1$ and $t_1$ defines the estimation of the secret key rate. A lower bound for the secret key rate requires an upper bound for $q_1$ and a lower bound for $t_1$. We have
\begin{equation}\label{Eqq1u}
q_1\leq q_1^{\rm U}=q-q_2^{\rm L}
\end{equation}
and [see Eq.~(\ref{Eqt1})]
\begin{equation}\label{Eqt1Lderiv}
t_1\geq \tilde t_1^{\rm L}=p_{\rm det}^{(1)}+
\tilde p^{(1),\rm L}_1\left(\frac1\eta-1\right).
\end{equation}
where 
\begin{equation}\label{Eqp11Lderiv}
p_1^{(1)}\geq \tilde p_1^{(1),\rm L}=
p_1-p_1^{(2),\rm U}-p_{\rm det}^{(3+)}.
\end{equation}
The replacement of $p_{\rm det}^{(3+)}$ by its upper bound $p_{\rm det}^{(3+),\rm U}$ in Ineq.~(\ref{Eqp11Lderiv}) gives the following bounds:
\begin{equation}\label{Eqp11L}
p_1^{(1)}\geq p_1^{(1),\rm L}=p_1-p_1^{(2),\rm U}-p_{\rm det}^{(3+),\rm U}
\end{equation}
and
\begin{equation}\label{Eqt1l}
t_1\geq t_1^{\rm L}=p_{\rm det}^{(1),\rm L}
+p_1^{(1),\rm L}\left(\frac1{\eta}-1\right).
\end{equation}

Also, we will use an upper bound on the two-photon contribution to the probability of detection $p_{\rm det}^{(2),\rm U}$. It is defined as the maximal value $p_{\rm det}^{(2)}$ such that $p_{\rm det}^{(1),\rm L}\geq0$ (under the assumption $p_{\rm det}>p_{\rm det}^{(3+),\rm U}$) and $\delta_x^{\rm L}\leq1$, where $\delta_x^{\rm L}$ is defined in Eq.~(\ref{EqDeltaxl}) below.

Let us recall that here we assume detector~1 to be the less efficient detector. In the opposite case, the roles of the outcomes 0 and 1 (in both bases) should be swapped over. In particular, the quantities $p_1$, $p_1^{(2),\rm U}$, and $p_1^{(1),\rm L}$ should be replaced by the corresponding quantities of the outcome~0.

\begin{theorem}\label{ThMain}
Suppose that $p_{\rm det}>p_{\rm det}^{(3+),\rm U}$ and $q_1^{\rm U}<t_1^{\rm L}/2$ for all $p_{\rm det}^{(2)}\in\left[0,p_{\rm det}^{(2),\rm U}\right]$. Then the secret key rate (\ref{EqDWmin}) subject to constraints (\ref{EqGamma1})--(\ref{EqGamma4}) is lower bounded by

\begin{equation}\label{EqMain}
K\geq \min_{p_{\rm det}^{(2)}} 
p^{(1),\rm L}_{\rm det}
\left[
1-h\left(\frac{1-\delta_x^{\rm L}}2\right)\right]
-p_{\rm det}h(Q_z),
\end{equation}
where
\begin{equation}\label{EqDeltaxl}
\delta_x^{\rm L}=
\frac{\sqrt\eta(t_1^{\rm L}-2q_1^{\rm U})}{p^{(1),\rm L}_{\rm det}}.
\end{equation}
The minimization is performed over the segment $p_{\rm det}^{(2)}\in\left[0,p_{\rm det}^{(2),\rm U}\right]$. The expression under minimization in Ineq.~(\ref{EqMain}) is a convex function of $p_{\rm det}^{(2)}$.
\end{theorem}

The requirement $q_1^{\rm U}<t_1^{\rm L}/2$ essentially means that the error rate is not too high.  If this condition is not met, Theorem~\ref{ThMain} cannot guarantee a positive secret key rate. If the condition $p_{\rm det}>p_{\rm det}^{(3+),\rm U}$ is not met, then our analysis cannot exclude the situation where all pulses arriving at the Bob's side contain three or more photons, which are treated as insecure. Hence, Theorem~\ref{ThMain} also cannot guarantee a positive secret key rate in this case.

\begin{remark}
It may seem counterintuitive that, in the estimation of the phase error rate (the argument of $h$ in Eq.~(\ref{EqMain})), we use the single-photon detection rate in the $z$ basis ($p^{(1)}_{\rm det}$), not in the $x$ basis. Though this is a rigorous result, we can give an intuition in favor of this: As we say in Remark~\ref{RemPhaseErr}, we have to estimate the phase error rate for the state attenuated in the $z$ basis used for key generation. This is the reason why the detection rate in the $z$ basis emerges in this estimate for the phase error rate.
\end{remark}

The calculations of the secret key rate with formula~(\ref{EqMain}) are presented on Fig.~\ref{Fig1}. For the simulation, we assume that
\begin{equation}\label{EqDep}
\rho_{AB}=(1-2Q)\ket{\Phi^+}\bra{\Phi^+}+2Q\frac{I_2}2\otimes\frac{I_2}2,
\end{equation}
where $0\leq Q\leq1/2$ and $I_2$ is an identity operator in the qubit space. In the case of perfect detection, $Q$ is the probability of error for both bases. The right-hand side of Eq.~(\ref{EqDep}) is the action of the depolarizing channel on the maximally entangled state given by Eq.~(\ref{EqPhi}). The depolarizing channel is a commonly used model for an actual transmission line (i.e., for the ``honest'' performance of the protocol, without  eavesdropping). The state given by Eq.~(\ref{EqDep}) gives $p_{\rm det}=(1+\eta)/2$, $p_1=\eta/2$, and $q=Q$. Also, we artificially set $p_{01}=10^{-5}$ due to dark counts. Strictly speaking, dark counts should be taken into account explicitly in the detection model, i.e., in Eqs.~(\ref{EqPOVM}) and~(\ref{EqPOVMtild}). This will be a subject for a future work. We emphasize that the security proof does not rely on a specific model (\ref{EqDep}) used for the simulation.

As an upper bound for the secret key rate, we can use a tight bound for a single-photon Bob's input, which was obtained obtained in Ref.~\cite{Bochkov} (a particular case of this bound was also obtained in Ref.~\cite{MaZhou}), see also Proposition~\ref{PropSingle} below with a simplified proof:
\begin{eqnarray}
K&=&p_{\rm det}\left[h\left(\frac{1-\delta_z}2\right)-
h\left(
\frac{1-\sqrt{\delta_x^2+\delta_z^2}}2
\right)\right]\nonumber\\
&-&p_{\rm det}h(Q_z),\label{EqK1tight}
\\
K&\geq&p_{\rm det}\left[1-
h\left(
\frac{1-\delta_x}2
\right)\right]
-p_{\rm det}h(Q_z).\label{EqK1}
\end{eqnarray}
where $\delta_x$ and $\delta_z$ are given by Eqs.~(\ref{EqDelta}) below. Formula~(\ref{EqK1tight}) is the exact value of the secret key rate in the case of the additional restriction $\rho_{AB}\in\mathfrak T(\mathbb C^2\otimes\mathbb C^2)$, i.e., if the Bob's input is  single photon. Formula~(\ref{EqK1}) gives a simplified and more rough bound. We use it (rather than Eq.~(\ref{EqK1tight})) for the derivation of Eq.~(\ref{EqMain}). From Fig.~\ref{Fig1}, we see that, in the practical case of small detection-efficiency mismatch ($\eta$ close to~1), formula (\ref{EqMain}) gives a secret key rate close to  both Eqs.~(\ref{EqK1tight}) and~(\ref{EqK1}).

\begin{figure}[t]
\begin{centering}
\includegraphics[width=1\columnwidth]{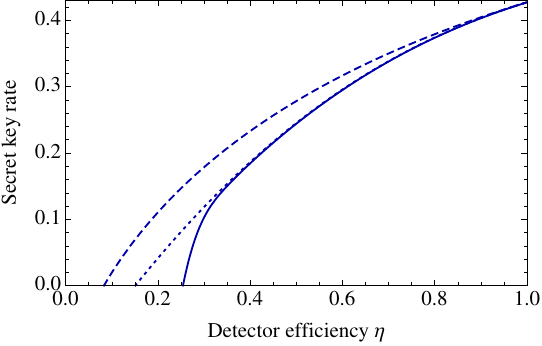}
\end{centering}
\vskip -4mm
\caption
{
Secret key of the BB84 protocol vs the efficiency of one of the detectors $\eta$. Another detector and the transmission line  are assumed to be perfect; otherwise, the secret key rate is reduced by a constant factor. The state $\rho_{AB}$ is given by Eq.~(\ref{EqDep}) with $Q=0.05$. This corresponds to $p_{\rm det}=(1+\eta)/2$, $p_1=\eta/2$, and $q=Q$. Also, we artificially set $p_{01}=10^{-5}$ due to dark counts. Solid line: formula (\ref{EqMain}). Dashed line: tight bound (\ref{EqK1tight}) for the single-photon case (used as an upper bound for the secret key rate). Dotted line: simplified (more rough) bound (\ref{EqK1}) for the single-photon case, which was used in the derivation of formula (\ref{EqMain}). We see that, for the practical case of small detection-efficiency mismatch ($\eta$ close to~1), formula (\ref{EqMain}) gives a secret key rate close to the upper bound.
}
\label{Fig1}
\end{figure}

The decrease of the secret key rate with the decrease of $\eta$ shown on Figs.~\ref{Fig1} is caused by two effects: the decrease of the average detector efficiency $(1+\eta)/2$ and detection-efficiency mismatch as such. To distinguish the influence of the mismatch as such, we compare the secret key rates for the mismatch case with the detector efficiencies 1 and $\eta$ and the no-mismatch case with both efficiencies equal to $(1+\eta)/2$. The secret key rate for the latter case is well-known and given by 
\begin{equation}\label{EqIdeal}
K=p_{\rm det}[1-2h(Q)]. 
\end{equation}
The ratio of the secret key rate in the mismatch case to that in the no-mismatch case is shown on Fig.~\ref{Fig2}. The solid line corresponds to the errorless case $Q=0$ and the dashed line corresponds to a high error rate $Q=0.09$. Recall that $Q\approx0.11$ is a maximal QBER for which the key distribution is possible ($1-2h(Q)>0$). The QBER $Q=0.09$ is high in the sense that it is close to the critical one $Q\approx0.11$.

We see that, first, the influence of mismatch on the secret key rate is larger for high QBERs and, second, if the mismatch is not very large, then the decrease of the secret key rate is also relatively small even for high QBERs. For example, the secret key rate for $\eta=0.8$ and $Q=0.09$ is above 90\% of the secret key rate for the no-mismatch case with the same $Q$ and average efficiency.

\begin{figure}[t]
\begin{centering}
\includegraphics[width=1\columnwidth]{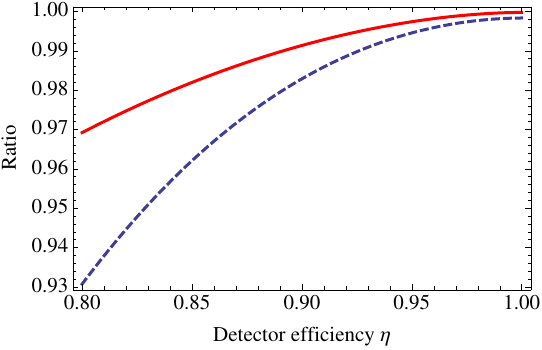}
\end{centering}
\vskip -4mm
\caption
{
Decrease of secret key rate in the detection efficiency-mismatch case with respect to the no-mismatch case: The ratio of the secret key rate in the mismatch case with the detector efficiencies 1 and $\eta$ to the secret key rate in the no-mismatch case with both  efficiencies equal to $(1+\eta)/2$. Solid line: no errors $Q=0$, dashed line: relatively high error rate $Q=0.09$ (close to the critical value for the case of perfect detection $Q\approx0.11$). All other parameters are the same as on Fig.~\ref{Fig1}. If the mismatch is small, then the decrease of secret key rate is  relatively small even for high QBERs.
}
\label{Fig2}
\end{figure}

\section{Proof of Theorem~\ref{ThMain}}\label{SecProof}

\subsection{General facts}

\begin{proposition}\label{PropSym}
Let $\Phi$ be a positive trace-preserving linear map acting on
$\mathfrak T(\mathcal H_A\otimes\mathcal H_B)$ that commutes with $\mathcal X\otimes{\rm Id}_{BE}$ and $\mathcal G$ and satisfies $\Phi^\dag(\Gamma_i)=\Gamma_i$ for all $i$. We assume that $p_{\rm det}$ given by (\ref{EqGamma1}) is included in the set of restrictions. Then
\begin{equation}\label{EqRestricH}
\sup_{\rho_{AB}\in\mathbf S}H(X|B)_{\tilde\rho'}
=
\sup_{\rho_{AB}\in\mathbf S'}H(X|B)_{\tilde\rho'},
\end{equation}
where $\mathbf S'=\mathbf S\cap{\rm Image}(\Phi)$.
\end{proposition}

Recall that the map $\mathcal X$ was given in Eq.~(\ref{EqZX}) and describes Alice's measurement in the $x$ basis. The dual map $\Phi^\dag$ is defined by the equality $\Tr[\Phi(\rho)A]=\Tr[\rho\,\Phi^\dag(A)]$ for an arbitrary trace-class operator $\rho$ and an arbitrary bounded operator $A$ \cite{Holevo}.

\begin{proof}
Denote $\rho'_{XB}=\mathcal X(\rho'_{AB})$ and, as before, $\tilde\rho'_{XB}=p_{\rm det}^{-1}\mathcal X(\rho'_{AB})$. We have
\begin{equation}\label{EqSymReducDeriv}
\begin{split}
-p_{\rm det}&H(X|B)_{\tilde\rho'}
=
p_{\rm det}D(\tilde\rho'_{XB}\|I_A\otimes\tilde\rho'_{B})\\
&=
D(\rho'_{XB}\|I_A\otimes\rho'_{B})\\
&=
D(\mathcal X(\mathcal G(\rho_{AB}))\|I_A\otimes\mathcal G(\rho_{B}))\\
&\geq
D(\Phi(\mathcal X(\mathcal G(\rho_{AB})))\|
\Phi(I_A\otimes\mathcal G(\rho_{B})))\\
&=
D(\mathcal X(\mathcal G(\Phi(\rho_{AB})))\|
I_A\otimes\mathcal G(\Phi(\rho_{B})))\\
&=-p_{\rm det}H(X|B)_{\Phi(\tilde\rho')}
\end{split}
\end{equation}
Here we have used monotonicity of quantum relative entropy under the action of a positive trace-preserving linear map on both arguments \cite{MullerHermesReeb}. Ineq.~(\ref{EqSymReducDeriv}) means that the substitution of $\rho_{AB}$ by $\Phi(\rho_{AB})$ does not decrease $H(X|B)$. Since $\Phi^\dag(\Gamma_i)=\Gamma_i$ for all $i$, $\Phi(\rho_{AB})$ satisfies all restrictions whenever $\rho_{AB}$ does. Also, we have used that $p_{\rm det}$ is also included in the set of restrictions and, hence, is not changed under the action of $\Phi$. Thus, for each $\rho_{AB}\in\mathbf S$, there exists $\Phi(\rho_{AB})\in\mathbf S'$ with the same or a greater value of the objective function. However, the supremum over $\mathbf S'$ cannot exceed the supremum over $\mathbf S$ because $\mathbf S'$ is a subset of $\mathbf S$. Hence, the supremum of $H(X|B)$ on $\mathbf S'$ coincides with that of on $\mathbf S$.
\end{proof}
Note that if $\Phi$ is a projector, i.e. $\Phi^2=\Phi$, then $\mathbf S'$ can be rewritten as
\begin{equation}\label{EqSprimeProj}
\mathbf S'=\mathbf S\cap\{\rho_{AB}\,\|\,\Phi(\rho_{AB})=\rho_{AB}\}.
\end{equation}
This means that the optimization can be performed only over the states that are invariant with respect to the map $\Phi$.

%Note that we always can include the identity operator into the family $\{U_g\}$.

%In the following, we will often omit the subindexes $AB$: $\rho\equiv\rho_{AB}$.

Let us define the  decoherence map with respect to the total number of photons: 
\begin{equation}\label{EqPhi1}
\Phi_1(\rho)=\sum_{n=0}^\infty\Pi_n\rho\Pi_n, 
\end{equation}
where
\begin{equation*}
\begin{split}
\Pi_n&=\sum_{k=0}^n
\ket{k,n-k}_z\bra{k,n-k}\\&=
\sum_{k=0}^n
\ket{k,n-k}_x\bra{k,n-k}.
\end{split}
\end{equation*}
$\Phi_1$ satisfies the conditions of Proposition~\ref{PropSym} for the constraints (\ref{EqGamma1})--(\ref{EqGamma4}) and $\Phi_1$ is a projector. Hence, we can take $\rho_{AB}$ invariant under this map: $\rho_{AB}=\Phi_1(\rho_{AB})$. Thus, without loss of generality, we assume that $\rho_{AB}$ and $\rho_{XB}=\mathcal X(\rho_{AB})$ are mixtures of  contributions with a certain number of photons: 

\begin{equation}\label{EqRhoSumN}
\rho_{XB}=\sum_{n=0}^\infty\rho^{(n)}_{XB},
\quad\rho_{XB}^{(n)}\in
\mathfrak T(\mathbb C^2\otimes(\mathbb C^2)^{\otimes_s n}).
\end{equation}
This is a formal proof of the observation in Ref.~\cite{EntanglVerif} that, without loss of generality, we can assume that a quantum non-demolition measurement of the  number of photons is performed before the actual measurement (\ref{EqPOVM}). Eq.~(\ref{EqRhoSumN}) allows us to analyze each $n$-photon part separately.

Since the state $\rho_{XB}$ is diagonal in the $x$ basis of Alice's space, we have
\begin{equation}\label{EqRhoXB}
\rho_{XB}=\frac12\left(
\ket+\bra+
\otimes\rho^{(+)}_B
+
\ket-\bra-
\otimes\rho^{(-)}_B\right),
\end{equation}
where $\rho^{(\pm)}_{XB}$ are again mixtures of contributions with a certain number of photons: 
\begin{equation}\label{EqRhoBsignSumN}
\rho^{(\pm)}_{XB}=\sum_{n=0}^\infty\rho^{(n,\pm)}_{XB},
\quad\rho_{XB}^{(n,\pm)}\in
\mathfrak T((\mathbb C^2)^{\otimes_s n}).
\end{equation}
In particular,
\begin{equation}\label{EqRhoXB1pre}
\rho_{XB}^{(1)}=\frac12\left(
\ket+\bra+
\otimes\rho^{(1,+)}_B
+
\ket-\bra-
\otimes\rho^{(1,-)}_B\right).
\end{equation}
Then, action of the map $\mathcal G$ on $\rho_{XB}$ can be decomposed into maps acting on the terms with a certain number of photons as follows:
\begin{equation}\label{EqRhoPrimeSumN}
\rho'_{XB}=\mathcal G(\rho_{XB})=
\sum_{n=1}^\infty G_n\rho_{XB}^{(n)}G_n,
\end{equation}
\begin{equation}\label{EqGn}
\begin{split}
G_n&=\begin{pmatrix}
1& & &\\
&\ddots& &\\
& & 1 &\\
& & & \sqrt{\theta_n}
\end{pmatrix}
\\
&=I_B-(1-\sqrt{\theta_n})\ket{0,n}_z\bra{0,n},
\end{split}
\end{equation}
where
the matrix representation is given in the basis $\{\ket{k,n-k}_z\}_{k=0}^n$ of the space $(\mathbb C^2)^{\otimes_s n}$. Now we can give formal definitions of the $n$-photon contributions from decompositions (\ref{EqDecompositions}) 
\begin{equation}\label{Eqqndecompose}
t_n=\Tr\rho^{(n)}_{XB},
\quad
p_{\rm det}^{(n)}=\Tr\Gamma_1\rho^{(n)}_{XB},
\quad
q_n=\Tr\Gamma_2\rho^{(n)}_{XB}.
\end{equation}

It turns out that
\begin{equation}
p_{\rm det}H(X|B)_{\tilde\rho'}
=\sum_{n=1}^\infty p_{\rm det}^{(n)}H(X|B)_{(n)},
\end{equation}
where the conditional entropy $H(X|B)_{(n)}$ is calculated for the state 
\begin{equation}\label{EqRhoXBn}
\tilde\rho^{(n)}_{XB}=\big(p_{\rm det}^{(n)}\big)^{-1}G_n\rho_{XB}^{(n)}G_n.
\end{equation} 
Due to (\ref{EqDWmin}), we should estimate $H(X|B)_{\tilde\rho'}$ from above.  Using the trivial upper bound $H(X|B)_{(n)}\leq1$ for $n\geq2$, we obtain
\begin{equation}\label{EqH1neq1}
p_{\rm det}H(X|B)_{\tilde\rho'}\leq
p_{\rm det}^{(1)}
H(X|B)_{(1)}+\big(p_{\rm det}-p_{\rm det}^{(1)}\big).
\end{equation}

Thus, our plan is to estimate $H(X|B)_{(1)}$ from above for a single-photon part $\rho^{(1)}_{XB}$ and to estimate the required parameters of $\rho^{(1)}_{XB}$ (for example, $p_{\rm det}^{(1)}$) using the known observables for the whole state $\rho_{AB}$.

\subsection{Single-photon part}\label{SecSingle}

An analytic formula of the tight bound for the single-photon case was  obtained in \cite{Bochkov}. Here we present a simplified proof of it.

\begin{proposition}\label{PropSingle}
For any $\rho^{(1)}_{XB}\in\mathfrak T(\mathbb C^2\otimes\mathbb C^2)$, $\rho_{XB}^{(1)}\geq0$, the following inequality is satisfied
\begin{equation}\label{Eqq1l}
\frac{q_1}{t_1}\geq\frac12-\frac1{t_1}\sqrt{\frac{p^{(1)}_0p^{(1)}_1}\eta},
\end{equation}
where $p_0^{(1)}=\Tr\rho^{(1)}_{B}P_0^{(z)}=p_{\rm det}^{(1)}-p_1^{(1)}$ (the probability of click of detector~0), and
\begin{equation}\label{EqHXB1tight}
H(X|B)_{(1)}\leq
1+h\left(
\frac{1-\sqrt{\delta_x^2+\delta_z^2}}2
\right)
-h\left(\frac{1-\delta_z}2\right),
\end{equation}
where
\begin{equation}
\delta_z=\frac{p^{(1)}_0-p^{(1)}_1}{p^{(1)}_{\rm det}},\quad
\delta_x=\frac{\sqrt\eta(t_1-2q_1)}{p_{\rm det}^{(1)}}.\label{EqDelta}
\end{equation}
The right-hand side of Ineq.~(\ref{EqHXB1tight}) is a non-increasing function of $|\delta_x|$ and $|\delta_z|$. In particular, a simplified formula (obtained by setting $\delta_z=0$) is true: 
\begin{equation}\label{EqHXB1}
H(X|B)_{(1)}\leq
h\left(
\frac{1-\delta_x}2
\right).
\end{equation}
\end{proposition}

\begin{proof}
Consider the problem of maximization of $H(X|B)_{(1)}$ over $\rho^{(1)}_{XB}\in\mathfrak T(\mathbb C^2\otimes\mathbb C^2)$ with only one constraint (\ref{EqGamma1}), where $\Gamma_1$ is substituted by its single-photon version $\Pi_1\Gamma_1\Pi_1$.
Consider the unitary transformation $Z\otimes Z$, where 
\begin{equation}
Z=\ket0\bra0-\ket1\bra1=\ket-\bra++\ket+\bra-
\end{equation} 
(phase flip), and the projector
\begin{equation}\label{EqPhi2}
\Phi_2(\rho^{(1)}_{XB})=\frac12
\left[\rho^{(1)}_{XB}+(Z\otimes Z)\rho^{(1)}_{XB}(Z\otimes Z)\right].
\end{equation}
$\Phi_2$ satisfies the conditions of Proposition~\ref{PropSym}. Hence, in view of Eq.~(\ref{EqSprimeProj}), we can restrict the set of states $\rho^{(1)}_{XB}$ to those that satisfy $\Phi_2(\rho^{(1)}_{XB})=\rho^{(1)}_{XB}$, or,
\begin{equation}\label{EqZsym}
\rho^{(1)}_{XB}=(Z\otimes Z)\rho^{(1)}_{XB}(Z\otimes Z).
\end{equation}
Application of Eq.~(\ref{EqZsym}) to Eq.~(\ref{EqRhoXB1pre}) gives 
\begin{equation}
\rho_B^{(1,-)}=Z\rho_B^{(1,+)}Z.
\end{equation}
Hence,
\begin{equation}\label{EqRhoXB1}
\begin{split}
G_1\rho_{XB}^{(1)}G_1=\frac12\Big(
&\ket+\bra+
\otimes G_1\rho^{(1,+)}_BG_1+\\
&\ket-\bra-
\otimes G_1\rho^{(1,-)}_BG_1
\Big),
\end{split}
\end{equation}
where 
\begin{equation}\label{EqG1rho1Z}
G_1\rho^{(1,-)}_BG_1=ZG_1\rho^{(1,+)}G_1Z
\end{equation}
(note that $G_1$ and $Z$ commute). 

Denote $\Phi_3(\rho^{(1)}_{XB})=[\rho^{(1)}_{XB}+(\rho^{(1)}_{XB})^*]/2$, where $*$ denotes the complex conjugation of the elements of $\rho^{(1)}_{XB}$ in the $z$ basis. This map satisfies the conditions of Proposition~\ref{PropSym}. Hence, without loss of generality, the matrix elements of $\rho^{(1)}_{XB}$ in the $z$ basis are assumed to be real.

Then, the matrix representation of the operator $G_1\rho^{(1,+)}_{B}G_1$ in the $z$ basis is
\begin{equation}\label{EqRhoXB1signz}
G_1\rho^{(1,+)}_{B}G_1=\frac{p_{\rm det}^{(1)}}2
\begin{pmatrix}
1+\delta_z&&\delta_x\\
\delta_x&&1-\delta_z
\end{pmatrix},
\end{equation}
where $\delta_z$ and $\delta_x$ are some real numbers satisfying the inequality
\begin{equation}\label{EqDeltaxz}
\delta_z^2+\delta_x^2\leq1.
\end{equation} Eq. (\ref{EqRhoXB1signz}) is a general form of a real positive-semidefinite $2\times2$ matrix with the trace $p_{\rm det}^{(1)}$. The matrix representation of the same operator in the $x$ basis is
\begin{equation}\label{EqRhoXB1signx}
G_1\rho^{(1,+)}_{B}G_1=\frac{p_{\rm det}^{(1)}}2
\begin{pmatrix}
1+\delta_x&&\delta_z\\
\delta_z&&1-\delta_x
\end{pmatrix}.
\end{equation}
We see that
\begin{equation}
\frac1{p_{\rm det}^{(1)}}\Tr\left[G_1\rho^{(1,+)}_BG_1
\big(\tilde P_0^{(b)}-\tilde P_1^{(b)}\big)\right]=
\delta_b,
\end{equation}
$b\in\{x,z\}$, which explains the denotations $\delta_z$ and $\delta_x$. In other words, $\delta_b$ is the difference between the probabilities of two outcomes in the case of perfect detection for the state $G_1\rho_B^{(1,+)}G_1$, i.e., for the state attenuated by the imperfect detection in the $z$ basis. Formula for $\delta_z$ in Eqs.~(\ref{EqDelta})  is now obvious. Also, it can be straightforwardly shown that
\begin{equation}
q_1=\Tr\big(\Gamma^{(1)}_2\rho_{XB}^{(1)}\big)=
\frac12\left(t_1-\frac{p_{\rm det}^{(1)}\delta_x}{\sqrt\eta}\right),
\end{equation} 
which gives formula for $\delta_x$ in Eqs.~(\ref{EqDelta}). 

After simple algebra, Ineq.~(\ref{Eqq1l}) follows from Eqs.~(\ref{EqDeltaxz}), (\ref{EqDelta}), and  
\begin{equation}
p_{\rm det}^{(1)}=p_0^{(1)}+p_1^{(1)},
\end{equation} 
which is a particular case of Eq.~(\ref{Eqpdetp}): $p_{0+01}^{(1)}=p_0^{(1)}$ since single-photon pulses do not lead to double clicks.

Let us calculate $H(X|B)_{(1)}$. It is equal to
\begin{equation}
H(X|B)_{(1)}=H(\tilde\rho_{XB}^{(1)})-H(\tilde\rho_{B}^{(1)}),
\end{equation}
where $\tilde\rho_{B}^{(1)}=\Tr_X\tilde\rho_{XB}^{(1)}$. In view of Eqs.~(\ref{EqG1rho1Z}) and~(\ref{EqRhoXB1signz}), we have
\begin{equation}\label{EqRhoXB1signzminus}
G_1\rho^{(1,-)}_BG_1=\frac{p_{\rm det}^{(1)}}2
\begin{pmatrix}
1+\delta_z&&-\delta_x\\
-\delta_x&&1-\delta_z
\end{pmatrix}.
\end{equation}
Hence, in view of Eqs.~(\ref{EqRhoXB1}), (\ref{EqRhoXB1signz}), and (\ref{EqRhoXB1signzminus}), the state $\tilde\rho_{XB}^{(1)}$ [see Eq.~(\ref{EqRhoXBn})] has two doubly degenerate eigenvalues $(1\pm\sqrt{\delta_z^2+\delta_x^2})/4$. Thus,
\begin{equation}
H(\tilde\rho_{XB}^{(1)})=1+h\left(
\frac{1-\sqrt{\delta_x^2+\delta_z^2}}2
\right),
\end{equation}
where $H(\rho)$ denotes the von Neumann entropy of the state $\rho$.
It turns out that
\begin{equation}
\tilde\rho^{(1)}_B=\frac{1}2
\begin{pmatrix}
1+\delta_z&&0\\
0&&1-\delta_z
\end{pmatrix}
\end{equation}
and 
\begin{equation}
H(\tilde\rho_{B}^{(1)})=h\left(
\frac{1-\delta_z}2
\right).
\end{equation}
Hence,
\begin{equation}\label{EqHXB1tighteq}
H(X|B)_{(1)}=
1+h\left(
\frac{1-\sqrt{\delta_x^2+\delta_z^2}}2
\right)
-h\left(\frac{1-\delta_z}2\right)
\end{equation}
for state (\ref{EqRhoXB1})--(\ref{EqRhoXB1signz}). Generally, we have Ineq.~(\ref{EqHXB1tight}) since, to obtain state  (\ref{EqRhoXB1})--(\ref{EqRhoXB1signz}) from the general state, we performed two positive maps $\Phi_2$ and $\Phi_3$, which, in general, increase $H(X|B)_{(1)}$. 

To proof the monotonicity of the right-hand side of Ineq.~(\ref{EqHXB1tight}) with respect to $|\delta_z|$ and $|\delta_x|$, we again (like in the proof of Proposition~\ref{PropSym}), express $H(X|B)_{(1)}$ in terms of the quantum relative entropy:
\begin{equation}\label{EqHXBD}
H(X|B)_{(1)}=-D(\tilde\rho_{XB}^{(1)}\|I_A\otimes \tilde\rho_{B}^{(1)}).
\end{equation}
Let us apply the completely positive map
\begin{equation}\label{EqDecohz}
\tilde\rho_{XB}^{(1)}\mapsto
\Big(1-\frac\alpha2\Big)\tilde\rho_{XB}^{(1)}
+\frac\alpha2(I\otimes Z)\tilde\rho_{XB}^{(1)}(I\otimes Z)
\end{equation}
or
\begin{equation}
\tilde\rho_{XB}^{(1)}\mapsto
\Big(1-\frac\alpha2\Big)\tilde\rho_{XB}^{(1)}
+\frac\alpha2(I\otimes X)\tilde\rho_{XB}^{(1)}(I\otimes X),
\end{equation}
$0\leq\alpha\leq1$, to both arguments of the relative entropy in (\ref{EqHXBD}). The relative entropy cannot increase or, equivalently, the conditional entropy cannot decrease under such action. Such maps corresponding to the substitutions
\begin{equation}
\delta_z\mapsto (1-\alpha)\delta_z
\end{equation}
and
\begin{equation}
\delta_x\mapsto (1-\alpha)\delta_x,
\end{equation}
respectively (i.e., to a partial decoherence in the bases $z$ and $x$, respectively). 

Hence, the right-hand side of Eq.~(\ref{EqHXB1tighteq}) is a non-decreasing function of $|\delta_z|$ and $|\delta_x|$. Setting $\alpha=1$ in map (\ref{EqDecohz}) gives $\delta_z=0$, corresponds to the full decoherence in the $z$ basis, and yields the right-hand side of Ineq.~(\ref{EqHXB1}).
\end{proof}

Thus, Eq.~(\ref{EqHXB1tight}) is a tight bound, while Eq.~(\ref{EqHXB1}) is a simplified and more rough one. In the latter case, $(1-\delta_x)/p^{(1)}_{\rm det}$ can be identified with the phase error rate, which plays a crucial role in  security proofs of QKD \cite{ShorPreskill,Koashi}. As in other cases of practical imperfections, it is not equal to the bit error rate for the $x$ basis \cite{GLLP}.

We will use a simplified formula (\ref{EqHXB1}) in our work since its deviation from a tight bound (\ref{EqHXB1tight}) is small for $\eta$ close to 1, which takes place in practice. The substitution of Ineq.~(\ref{EqHXB1}) into Ineq.~(\ref{EqH1neq1}) gives
\begin{equation}\label{EqH1neq1new}
p_{\rm det}H(X|B)_{\tilde\rho'}\leq
p_{\rm det}^{(1)}
h\left(
\frac{1-\delta_x}2
\right)
+\big(p_{\rm det}-p_{\rm det}^{(1)}\big).
\end{equation}
We need to estimate $\delta_x$ from below. Hence, in view of Eq.~(\ref{EqDelta}), we should estimate $t_1$ from below and $q_1$ from above. The estimation of $t_1$ will include $p_{\rm det}^{(1)}$ (see Eq.~(\ref{Eqt1l})), hence we postpone the investigation of the dependence of Ineq.~(\ref{EqH1neq1new}) on $p_{\rm det}^{(1)}$ until we obtain all estimates in Sec.~\ref{SecFin}.

\begin{remark}\label{RemSym}
In Eq.~(\ref{EqPhi2}), we used a special form of projector. More generally, consider a unitary representation  $\{U_g\}_{g\in\mathbf G}$ of a finite group $\mathbf G$ of the order $|\mathbf G|$, where each unitary operator $U_g$ acts on $\mathcal H_A\otimes\mathcal H_B$. If $\Phi_g(\rho_{AB})=U_g\rho_{AB}U_g^\dag$ for all $g$ satisfy the conditions of Proposition~\ref{PropSym}, then Eq.~(\ref{EqRestricH}) with
\begin{equation}\label{EqSprimeU}
\mathbf S'=\mathbf S\cap\{\rho_{AB}\,\|\,U_g\rho_{AB}U_g^\dag=\rho_{AB}
\text{ for all } g\}
\end{equation}
holds. So, the optimization can be performed only over the states that are invariant with respect to all unitary maps. Indeed, consider the map
\begin{equation}\label{EqPhiMean}
\Phi(\rho_{AB})=\frac1{|\mathbf G|}\sum_{g\in\mathbf G}U_g\rho_{AB}U_g^\dag.
\end{equation}
It satisfies the conditions of Proposition~\ref{PropSym}. Since $\mathbf G$ is a group, the right-hand side of Eq.~(\ref{EqPhiMean}) is invariant under the action of each $U_g$. Hence, $\Phi$ is a projector and  Eq.~(\ref{EqRestricH}) with $\mathbf S'$ given by Eq.~(\ref{EqSprimeProj}) holds.  From the other side, the invariance under the action of each $U_g$ yields the invariance under the action of $\Phi$. Hence, Eqs.~(\ref{EqSprimeProj}) and (\ref{EqSprimeU}) are equivalent for this choice of $\Phi$. 

We will not use general formula (\ref{EqSprimeU}) here, but it can be useful in other problems: If a problem has a certain symmetry, then it can be used to restrict the search space to those that obey the same symmetry. This observation about symmetry groups was described in Ref.~\cite{FL}. In the presented framework, this is a particular case of Proposition~\ref{PropSym}, which also allows for a general positive and trace-preserving map, not necessarily an average over a unitary representation of a group.
\end{remark}

\subsection{Three- and more photon part}\label{SecTriple}

In this subsection, we prove Ineq.~(\ref{Eqpdet3u}). Denote the double click probabilities in the case of the perfect and imperfect detections:
\begin{equation}
\begin{split}
\tilde p_{01}^{(b)}&=\Tr\rho_{AB}\big(I_A\otimes\tilde P^{(b)}_{01}\big),\\
p_{01}^{(b)}&=\Tr\rho_{AB}\big(I_A\otimes P^{(b)}_{01}\big),\\
\end{split}
\end{equation}
Define also the corresponding double click probabilities averaged over two bases (the mean double click probabilities):
\begin{equation}
\tilde p_{01}=\frac{\tilde p_{01}^{(z)}+\tilde p_{01}^{(x)}}2,\qquad
p_{01}=\frac{p_{01}^{(z)}+p_{01}^{(x)}}2,
\end{equation}
$b\in\{z,x\}$. In force of
\begin{equation}\label{EqP01ineq}
\tilde P_{01}^{(b)}\geq P_{01}^{(b)}\geq\eta\tilde P_{01}^{(b)},
\end{equation}
we have 
\begin{equation}\label{Eqp01ineq}
\tilde p_{01}\geq p_{01}\geq\eta\tilde p_{01}.
\end{equation} 
The first inequality in Ineq.~(\ref{EqP01ineq}) means that the imperfect detection turns some double clicks into single clicks of detector~0. So, some double clicks are lost in the case of imperfect detection. The second inequality means that the fraction of such ``lost'' double clicks is at most $1-\eta$.

\begin{proposition}\label{Propdb}
Consider the state $\rho_{XB}$ of form (\ref{EqRhoSumN}).
Then 
\begin{equation}\label{Eqt3u}
t_{3+}\leq \frac{p_{01}}{\eta p_{01}^{\rm min}}\equiv t_{3+}^{\rm U},
\end{equation}
where $t_{3+}=\sum_{n=3}^\infty t_n$ and $p_{01}^{\rm min}$ is defined by Eq.~(\ref{Eqp01min}).
\end{proposition}

We begin with two lemmas.

\begin{lemma}\label{Lemdbmain}
Suppose $\rho_{B}\in\mathfrak T((\mathbb C^2)^{\otimes_s n})$ for $n\geq 3$ and $\rho_{B}$ is a density operator, i.e., $\rho_B\geq0$ and $\Tr\rho_B=1$. Then
\begin{equation}\label{EqDbtilde}
2\tilde p_{01}\log(2^{n-1}-1)+2h(\tilde p_{01})\geq n-2.
\end{equation}

\end{lemma}

In particular, it follows that $\tilde p_{01}$ is strictly positive for $n\geq3$.

\begin{proof}

The Hilbert space $(\mathbb C^2)^{\otimes_s n}$ is naturally embedded into the space $(\mathbb C^2)^{\otimes n}$. Let us then consider the operators $\rho_B$ and $\tilde P^{(b)}_{01}$ as  operators acting on $(\mathbb C^2)^{\otimes n}$. For a binary string $a=a_1\ldots a_n$, denote $\ket{a_1\ldots a_n}=\ket{a_1}\otimes\ldots\otimes\ket{a_n}\in(\mathbb C^2)^{\otimes n}$. Let $|a|$ denote the Hamming weight of $a$, i.e., the number of ones in $a$. The embedding is as follows:
\begin{equation}
(\mathbb C^2)^{\otimes_s n}\ni\ket{n-k,k}_z\mapsto
\frac1{\sqrt{k!}}\sum_{\substack{a\in\{0,1\}^n\\|a|=k}}\ket a
\in(\mathbb C^2)^{\otimes n}
\end{equation}
Define also the operators corresponding to a double click in $(\mathbb C^2)^{\otimes n}$ (in the case of the perfect detection):
\begin{equation}
\begin{split}
\bar P^{(z)}_{01}&=\sum_{\substack{a\in\{0,1\}^n\\1\leq|a|\leq n-1}}
\ket a\bra a,\\
\bar P^{(x)}_{01}&=\sum_{\substack{a\in\{0,1\}^n\\1\leq|a|\leq n-1}}
H^{\otimes n}\ket a\bra a H^{\otimes n}.
\end{split}
\end{equation} 
The operators $\bar P^{(b)}_{01}$ differ from the embeddings of $\tilde P^{(b)}_{01}$. However, if $\rho_{B}\in\mathfrak T((\mathbb C^2)^{\otimes_s n})$, i.e., $\rho_B$ is symmetric with respect to the permutations of the qubits, then
\begin{equation}
\Tr\tilde P^{(b)}_{01}\rho_{B}=\Tr\bar P^{(b)}_{01}\rho_{B}.
\end{equation}
Hence, $\tilde p_{01}^{(b)}$  can be equivalently defined in terms of $\bar P^{(b)}_{01}$. Let $Z$ and $X$ be the usual Pauli operators, and let $H(Z^{\otimes n})$ and $H(X^{\otimes n})$ be the Shannon entropies of the results of the measurements of the $n$-qubit observables $Z^{\otimes n}$ and $X^{\otimes n}$, respectively. Then
\begin{equation}\label{EqHZ}
\begin{split}
H(Z^{\otimes n})&=
h\left(\tilde p^{(z)}_{01}\right)+
\tilde p^{(z)}_{01}H(Z^{\otimes n}|\text{ double click})
\\&+
\left(1-\tilde p^{(z)}_{01}\right)H(Z^{\otimes n}|\text{ single click})
\\&
\leq h\left(\tilde p^{(z)}_{01}\right)+\tilde p^{(z)}_{01}\log(2^n-2)+1-\tilde p^{(z)}_{01}
\\&=
h\left(\tilde p^{(z)}_{01}\right)+\tilde p^{(z)}_{01}\log(2^{n-1}-1)+1.
\end{split}
\end{equation}
The inequality takes place since the single click event corresponds to two outcomes ($a=0\ldots0$ and $a=1\ldots1$) and the double click corresponds to the rest $2^n-2$ outcomes. Analogously,
\begin{equation}\label{EqHX}
H(X^{\otimes n})
\leq h\left(\tilde p^{(x)}_{01}\right)+\tilde p^{(x)}_{01}\log(2^{n-1}-1)+1.
\end{equation}
From the other side, due to entropy uncertainty relations \cite{MaassenUffink}, we have
\begin{equation}\label{EqEUR}
H(Z^{\otimes n})+H(X^{\otimes n})\geq n.
\end{equation}
Using Ineqs.~(\ref{EqHZ})--(\ref{EqEUR}) and the concavity of $h$, we obtain
\begin{multline}
2\tilde p_{01}\log(2^{n-1}-1)+2h(\tilde p_{01})
\\\geq 
\left(\tilde p^{(z)}_{01}+\tilde p^{(x)}_{01}\right)\log(2^{n-1}-1)+
h\left(\tilde p^{(z)}_{01}\right)+h\left(\tilde p^{(x)}_{01}\right)\\\geq
n-2,
\end{multline}
q.e.d.
\end{proof}

\begin{lemma}\label{Lemdbmonot}
Suppose $n\geq 3$. There is a unique value $p^{{\rm min}, (n)}_{01}$ for $p^{(n)}_{01}$ that turns inequality~(\ref{EqDbtilde}) into the equality. This is a lower bound for the mean double click probability for a given $n$. Moreover, $p^{{\rm min}, (n)}_{01}$ is a non-decreasing function of $n$ for $n\geq3$.
\end{lemma}

Note that $p_{01}^{\rm min}$ defined in Eq.~(\ref{Eqp01min}) coincides with $p^{{\rm min}, (3)}_{01}$. The proof is technical and does not contain essential ideas; so, it is given in Appendix~\ref{AppLemdbmonot}.

\begin{proof}[Proof of Proposition~\ref{Propdb}]
In force of Lemmas~\ref{Lemdbmain} and~\ref{Lemdbmonot}, 
\begin{equation}
\frac1{2t_n}\Tr\rho_{XB}^{(n)}[\tilde P_{01}^{(z)}+\tilde P_{01}^{(x)}]\geq
p^{{\rm min},(n)}_{01}\geq p^{\rm min}_{01}
\end{equation}
for $n\geq 3$, where, as before, $t_n=\Tr\rho_{XB}^{(n)}$. Note that $p^{\rm min}_{01}$, which is defined by Eq.~(\ref{Eqp01min}), is equal to $p^{{\rm min},(3)}$.   Due to Ineq.~(\ref{EqP01ineq}),
\begin{equation}
\frac1{2t_n}\Tr\rho_{XB}^{(n)}[P_{01}^{(z)}+P_{01}^{(x)}]\geq
\eta p^{{\rm min}}_{01}.
\end{equation}
Then,
\begin{equation}
p_{01}=\frac12\sum_{n=3}^\infty\Tr\rho_{XB}^{(n)}[P_{01}^{(z)}+P_{01}^{(x)}]\geq
\eta t_{3+} p^{\rm min}_{01}.
\end{equation}
The proposition has been proved.
\end{proof}

\subsection{Two-photon part}\label{SecDouble}

Unfortunately, bound (\ref{EqDbtilde}) is trivial for the case $n=2$. As we mentioned in Sec.~\ref{SecMainTh}, the two-photon Bell state $\ket{\Phi^+}$ (\ref{EqBellPhi}) produces no double clicks in both bases. So, estimation of the fraction of two-photon pulses requires a separate analysis. In this subsection, we prove Ineqs.~(\ref{Eqq2L}) and~(\ref{Eqp12U}).

Let us repeat the intuition behind the analysis which was already mentioned in Sec.~\ref{SecMainTh}. If the number of double clicks is small, then the two-photon part of the Bob's state is close to the pure state $\ket{\Phi^+}$. So, the Bob's subsystem is almost uncorrelated with the Alice's one and the QBER is close to 1/2. By this reason, as we will see in the next subsection, it is not advantageous for Eve to increase the fraction of two-photon part of the Bob's state: she introduces a large number of errors obtaining only a small amount of information about the key.

\begin{proposition}\label{PropQ2}
The following inequalities for the state $\rho_{XB}$ of form (\ref{EqRhoSumN})--(\ref{EqRhoBsignSumN}) hold: 
\begin{eqnarray}
\frac{q_2}{t_2}&\geq&\frac{1+\theta_2/\eta}4-\frac{\theta_2}\eta\sqrt{\frac{2p_{01}}{\eta t_2}},\label{Eqq2}\\
\frac{p_1^{(2)}}{t_2}&\leq& 
\frac{\theta_2}2+\theta_2\sqrt{\frac{2p_{01}}{\eta t_2}},\label{Eqp12}
\end{eqnarray}
\end{proposition}

\begin{proof}
First consider the case $t_2=1$. The vectors $\ket{\Phi^+}$, $\ket{\Phi^-}$, and $\ket{\Psi^+}$, where
\begin{equation}
\begin{split}
\ket{\Phi^-}&=\frac1{\sqrt2}(\ket{2,0}_z-\ket{0,2}_z)=\ket{1,1}_x,\\
\ket{\Psi^+}&=\ket{1,1}_z=\frac1{\sqrt2}(\ket{2,0}_x-\ket{0,2}_x),
\end{split}
\end{equation}
compose a basis in $(\mathbb C^2)^{\otimes_s2}$.
As we see, the  state $\ket{\Phi^-}\bra{\Phi^-}$ produces $\tilde p_{01}^{(z)}=0$ and $\tilde p_{01}^{(x)}=1$, while the state $\ket{\Psi^+}\bra{\Psi^+}$ produces $\tilde p_{01}^{(z)}=1$ and $\tilde p_{01}^{(x)}=0$. In both cases, $\tilde p_{01}=1/2$. It follows that
\begin{equation}
\tilde p_{01}= \frac{1-\braket{\Phi^+|\rho^{(2)}_{B}|\Phi^+}}2.
\end{equation}
Then,
\begin{equation}\label{EqTF}
T(\rho^{(2)}_B,\ket{\Phi^+}\bra{\Phi^+})
\leq\sqrt{1-\braket{\Phi^+|\rho^{(2)}_{B}|\Phi^+}}=\sqrt{2\tilde p_{01}},
\end{equation}
where $T$ is the trace distance, and we have used the well-known relation between the trace distance and the fidelity.

If $\rho^{(2)}_{XB}=\frac12I_A\otimes\ket{\Phi^+}\bra{\Phi^+}$, then $q_2=(1+\theta_2/\eta)/4$ and $p_1^{(2)}=\theta_2/2$.

Recall that, for arbitrary density operators $\rho$ and $\sigma$, the trace distance is proportional to the sum of the absolute values of the operator $\rho-\sigma$. Due to general form (\ref{EqRhoXB}) of $\rho_{XB}^{(2)}$,
\begin{multline}
\rho_{XB}^{(2)}-\frac12I_A\otimes\ket{\Phi^+}\bra{\Phi^+}\\=
\frac12\Big[
\ket+\bra+
\otimes(\rho^{(+)}_B-\ket{\Phi^+}\bra{\Phi^+})
\\+
\ket-\bra-
\otimes(\rho^{(-)}_B-\ket{\Phi^+}\bra{\Phi^+})\Big].
\end{multline}
Obviously, the spectrum of this operator is the union of the spectra of the operators $(\rho^{(\pm)}_B-\ket{\Phi^+}\bra{\Phi^+})/2$. Hence,
\begin{multline}
T\left(\rho_{XB}^{(2)},\frac12I_A\otimes\ket{\Phi^+}\bra{\Phi^+}\right)
\\=
\frac12T\left(\rho^{(2,+)}_{B},\ket{\Phi^+}\bra{\Phi^+}\right)
+
\frac12T\left(\rho^{(2,-)}_{B},\ket{\Phi^+}\bra{\Phi^+}\right).
\end{multline}
Now we can use the relation between the trace distance and the fidelity [already used in Ineq.~(\ref{EqTF})] to obtain
\begin{multline}
T\left(\rho_{XB}^{(2)},\frac12I_A\otimes\ket{\Phi^+}\bra{\Phi^+}\right)
\\
\leq
\frac12\sqrt{1-\braket{\Phi^+|\rho^{(2,+)}_{B}|\Phi^+}}
+
\frac12\sqrt{1-\braket{\Phi^+|\rho^{(2,-)}_{B}|\Phi^+}}.
\end{multline}
Due to the concavity of the square root and since 
\begin{equation}
\frac{\rho^{(2,+)}_{B}+\rho^{(2,-)}_{B}}2=\Tr_X\rho_{XB}^{(2)}=\rho_B^{(2)},
\end{equation}
we have
\begin{equation}
T\left(\rho^{(2)}_{XB},\frac12I_A\otimes\ket{\Phi^+}\bra{\Phi^+}\right)
\leq
\sqrt{1-\braket{\Phi^+|\rho^{(2)}_{B}|\Phi^+}}.
\end{equation}
Finally, Ineqs.~(\ref{EqTF}) and~(\ref{Eqp01ineq}) give
\begin{equation}
T\left(\rho_{XB}^{(2)},\frac12I_A\otimes\ket{\Phi^+}\bra{\Phi^+}\right)
\leq\sqrt{2\tilde p_{01}}\leq\sqrt{2p_{01}/\eta}.
\end{equation}

Recall the known properties of  trace distance: $|\Tr[P(\rho-\sigma)]|\leq T(\rho,\sigma)$ for an arbitrary projector $P$ and, hence, $|\Tr[A(\rho-\sigma)]|\leq \|A\|_\infty T(\rho,\sigma)$ for an arbitrary self-adjoint operator $A$ with the operator norm $\|A\|_\infty$. Since $\|\Pi_2\Gamma_2\Pi_2\|_\infty=\theta_2/\eta$ and $\|\Pi_2\Gamma_3\Pi_2\|_\infty=\theta_2$, we have
\begin{equation*}
\begin{split}
\left|q_2-\frac{1+\theta_2/\eta}4\right|&\leq
\frac{\theta_2}\eta
\sqrt{\frac{2p_{01}}\eta},
\\
\left|p_1^{(2)}-\frac{\theta_2}2\right|&\leq
\theta_2
\sqrt{\frac{2p_{01}}\eta},
\end{split}
\end{equation*} 
hence,
\begin{eqnarray}
&&q_2\geq\frac{1+\theta_2/\eta}4-
\frac{\theta_2}\eta\sqrt{\frac{2p_{01}}\eta},\\
&&p_1^{(2)}
\leq\frac{\theta_2}2+
\theta_2\sqrt{\frac{2p_{01}}{\eta}}.
\end{eqnarray}
We have obtained the required bound (\ref{Eqq2}) for the case $t_2=1$.
If $t_2<1$, then all derivations above should be performed for $\rho^{(2)}_{XB}/t_2$. In particular, the observables $q_2$, $p^{(2)}_1$, and $p_{01}$ should be substituted by $q_2/t_2$, $p_1^{(2)}/t_2$, and $p_{01}/t_2$. Thus, we obtain Eqs.~(\ref{Eqq2}) and~(\ref{Eqp12}).
\end{proof}

The required Ineqs.~(\ref{Eqq2L}) and~(\ref{Eqp12U}) obviously follow from Ineqs.~(\ref{Eqq2}), (\ref{Eqp12}), and~(\ref{Eqtnineq}).

\subsection{Final formula and remarks}\label{SecFin}

We proceed to the final steps of the proof of Theorem~\ref{ThMain}. 
From Ineqs.~(\ref{EqDWmin}) and (\ref{EqH1neq1new}), we have
\begin{equation}\label{EqMainPre}
K\geq \min_{(p_{\rm det}^{(2)},p_{\rm det}^{(3+)})}
p^{(1)}_{\rm det}
\left[
1-h\left(\frac{1-\tilde\delta_x^{\rm L}}2\right)\right]
-p_{\rm det}h(Q_z),
\end{equation}
where
\begin{equation}\label{EqDeltaxlderiv}
\tilde\delta_x^{\rm L}=
\frac{\sqrt\eta(\tilde t_1^{\rm L}-2q_1^{\rm U})}{p^{(1)}_{\rm det}}
\end{equation}
and
\begin{equation}\label{Eqpdet1}
p_{\rm det}^{(1)}=p_{\rm det}-p_{\rm det}^{(2)}-p_{\rm det}^{(3+)}\geq0.
\end{equation}
The minimization in Ineq.~(\ref{EqMainPre}) is performed over all $\big(p_{\rm det}^{(2)},p_{\rm det}^{(3+)}\big)$ that can be obtained from some density operator $\rho_{AB}$ and given constraints (\ref{EqGamma1})--(\ref{EqGamma4}). To estimate this expression from below, we will minimize over, generally, a broader range: over the pairs of positive numbers $\big(p_{\rm det}^{(2)},p_{\rm det}^{(3+)}\big)$ such that $\tilde\delta_x^{\rm L}\leq 1$ and restrictions (\ref{Eqpdet1}) and (\ref{Eqpdet3u}) are satisfied. Denote this set $D$. Restrictions (\ref{Eqpdet1}) and  $\tilde\delta_x^{\rm L}\leq 1$ come from the positivity of $\rho_{AB}$ [cf. Ineq.~(\ref{EqDeltaxz})], while Ineq.~(\ref{Eqpdet3u}) also has been proved. In other words, any density operator $\rho_{AB}$ satisfies these restrictions.

Replacement of $p_{\rm det}^{(3+)}$ by the upper bound $p_{\rm det}^{(3+),\rm U}$ in Eq.~(\ref{Eqp11Lderiv}) and, consequently, in Eqs.~(\ref{Eqt1Lderiv}) and (\ref{EqDeltaxlderiv}) gives $p_1^{(1),\rm L}$, $t_1^{\rm L}$, and $\delta_x^{\rm L}$, respectively [see Eqs.~(\ref{Eqp11L}), (\ref{Eqt1l}), and (\ref{EqDeltaxl})]. This turns Ineq.~(\ref{EqMainPre}) into the desired Ineq.~(\ref{EqMain}). But  this replacement should be justified.

\begin{lemma}\label{Lempdet3}
Under the conditions of Theorem~\ref{ThMain}, the minimum of the right-hand side of Ineq.~(\ref{EqMainPre}) over $\big(p_{\rm det}^{(2)},p_{\rm det}^{(3+)}\big)\in D$ is equal to the right-hand side of Ineq.~(\ref{EqMain}). The expression under minimization in Ineq.~(\ref{EqMain}) is well-defined for all $p_{\rm det}^{(2)}\in\big[0,p_{\rm det}^{(2),\rm U}]$.
\end{lemma}

The proof is technical and is given in Appendix~\ref{Apppdet3}. 

To finish the proof of Theorem~\ref{ThMain}, we need to prove that the expression under minimization in Ineq.~(\ref{EqMain}) is a convex function of $p_{\rm det}^{(2)}$. If we redenote $p_{\rm det}^{(1),\rm L}$ as an independent variable $x$ (not to be confused with the denotation of a basis), then $p_{\rm det}^{(2)}$ can be expressed as $p_{\rm det}^{(2)}=p_{\rm det}-p_{\rm det}^{(3+),\rm U}-x$, then we can equivalently minimize over $x$. Since $t_1^{\rm L}-2q_1^{\rm U}$ is a convex function of $x$, it suffices to prove the following lemma.

\begin{lemma}\label{LemConcave}
Consider the function of the form
\begin{equation}
f(x)=x h\left(\frac12-\frac{g(x)}x\right)
\end{equation}
defined on some segment $x\in[x_0,x_1]$, $x_0,x_1\geq0$, where $0\leq\frac{g(x)}x\leq\frac12$. If $g(x)$ is convex, then $f(x)$ is concave.
\end{lemma}
The proof is also technical and is given in Appendix~\ref{AppConcave}. This finishes the proof of Theorem~\ref{ThMain}.

\begin{remark}\label{RemImpr}
As we see on Fig.~\ref{Fig1}, formula~(\ref{EqMain}) gives a significant deviation from the upper bound~(\ref{EqK1tight}) for a large detection-efficiency mismatch (small $\eta$). For $\eta<0.3$, the formula~(\ref{EqMain}) gives a worse result even in comparison with the simplified bound~(\ref{EqK1}) for the single-photon Bob's input. Though, usually the mismatch is not large, here we discuss a way to improve bound~(\ref{EqMain}). Namely, we can use a more precise estimate of $p_1^{(2)}$ instead of Ineq.~(\ref{Eqp12U}). Indeed, in Ineq.~(\ref{Eqp12U}) (compare with Ineq.~(\ref{Eqp12})),  we simply use the bound $t_2\leq p_{\rm det}^{(2)}/\theta_2$. However, we can estimate $t_2$ and $p_1^{(2)}$ using  Ineq.~(\ref{Eqp12}) and the equation
\begin{equation}
t_2=p_{\rm det}^{(2)}+p_1^{(2)}\left(\frac1{\theta_2}-1\right)
\end{equation}
[analogous to Eq.~(\ref{Eqt1}), a consequence of Eqs.~(\ref{Eqpdetp}) and~(\ref{Eqtp})].
Then the estimation of $t_2$ and $p_2$ is reduced to a solution of a quadratic equation.
\end{remark}

\begin{remark}\label{RemDark}
In the case of no dark counts, we could take the $\Gamma_2$ observable (weighted mean erroneous detection rate in the $x$ basis) as
\begin{equation}\label{EqGamma2Simp}
\Gamma_2=
\eta^{-1}\ket+\bra+\otimes
P^{(x)}_1+
\ket-\bra-\otimes
P^{(x)}_0
\end{equation}
instead of Eq.~(\ref{EqGamma2}),
i.e., do not include the double clicks. Indeed, we used a non-trivial estimation of $H(X|B)$ (which uses the observable $\Gamma_2$) only for the single-photon part. In practice, double clicks may originate either from multiphoton pulses on Bob's side or from dark counts. We have only the latter origin of double clicks if we consider the single-photon pulses. Since we did not include the dark counts in the detection model [Eqs.~(\ref{EqPOVM}) and~(\ref{EqPOVMtild})], single-photon pulses do not cause double clicks in our model and the simplified formula (\ref{EqGamma2Simp}) could be used. But since, in practice, there are dark counts (though not taken into account in our model), more general  formula~(\ref{EqGamma2}) should be used. Rigorous inclusion of dark count rates into the detection model is a subject for a future work.
\end{remark}

\begin{remark}
Let us discuss the tightness of our bound (\ref{EqP01ineq}) for the three- and more photon part of the density operator. We have rigorously proved that the mean double click rate for every three- or more photon state is at least $p_{01}^{\rm min}\approx0.06$. From the tight numerical bounds of Ref.~\cite{Zhang}, it follows that this bound is approximately 0.25 (under the conjecture that the minimal mean double click rate is a non-decreasing function of a photon number). The difference between the tight numerical result and our analytic bound is larger for the case of efficiency mismatch.

In the case of the normal operation of a QKD system (no eavesdropping or eavesdropping that does not change the statistics of detections), a double click is a rare event because it occurs in the case of simultaneous occurrence of two low-probability events: a dark count and a detection of a photon after the transmission loss. In this case, as we see on the plots, the derived analytic bound (\ref{EqMain}) for the secret key rate is very close to the corresponding single-photon result: The influence of the multiphoton part is anyway small. Since the actual double click rate is low, the difference between bound (\ref{EqP01ineq}) and the tight numerical bound is not critical.

In Ref.~\cite{Zhang}, an example when Eve artificially resends multiple photons is considered. In this case, double click may be not a rare event and the numerical method may give essentially better results. But the aim of the present paper is a good bound for the normal operation of a QKD system.
\end{remark}

\begin{remark}\label{RemNonLinDet}
As we discussed before, Eq.~(\ref{EqPOVMtild}) for the POVM corresponding to an imperfect measurement is valid if the imperfect detection can be modeled by an asymmetric beam splitter followed by a perfect detection. This is an approximation. However, as we can observe, the analysis was  relied neither on the precise formula (\ref{EqPOVMtild}) for the POVM nor on the precise formula (\ref{EqTheta}) for $\theta_n$. Only inequality (\ref{EqP01ineq}) is essential. However, expressions~(\ref{Eqpdetp}) and~(\ref{Eqtp}) and, as a consequence, inequality (\ref{Eqtnineq}) essentially rely on the fact that one detector has the perfect efficiency. So, possibility of a reduction to this case is important. In other words, we assume that detector 0 (the detector with a larger efficiency) can be modeled by an asymmetric beam splitter followed by a perfect detection, but detector 1 is not assumed to be equivalent to this model. If we do not know the precise value of $\theta_2$ used in Proposition~\ref{PropQ2}, then we can use the bounds $\eta\leq\theta_2\leq1$ instead, which means that the probability of detection of a two-photon signal is at least as large as the probability of detection of a single-photon signal.
\end{remark}

\section{Decoy state method in the case of detection-efficiency mismatch}\label{SecDecoy}

\subsection{Estimations}\label{SecDecoyEst}

Now let us take into account that Alice sends not true single-photon pulses, but weak coherent pulses (with the randomized phase). We consider the scheme of one signal state and two weak decoy states. This means that each Alice's pulse can be either a signal pulse with the intensity $\mu_{\rm s}=\mu$ (used for key generation) or one of two decoy pulses with the intensities $\mu_{\rm d_1}=\nu_1$ and $\mu_{\rm d_2}=\nu_2$, with the conditions $0\leq\nu_2<\nu_1$ and $\nu_1+\nu_2<\mu$.

We follow the method of Ref.~\cite{MaLo2005}, where a lower bound for the number of detections originated from the single-photon pulses and an upper bound for the error rate for the single-photon pulses were derived.

The decoy state method for the case of detection-efficiency mismatch was developed in Ref.~\cite{Bochkov}. It is observed there that the decoy state estimates of Ref.~\cite{MaLo2005} have a nice feature that they actually do not impose any assumptions on the efficiency of detectors. They are based solely on simple counting of detections and erroneous detections. In principle, the decoy state estimations are still valid even if Eve has full control on the detector efficiencies. In the case of ideal devices, one needs to estimate only two quantities: the number of single-photon pulses registered by Bob and the error rate in the set of the registered single-photon pulses (note that, in this section, ``single-photon pulse'' and ``multiphoton pulse'' refer to Alice's side). As stated above, in the case of detection efficiency mismatch, these estimates are still valid. However, more detailed information is required in this case, not just these two quantities. In Ref.~\cite{Bochkov}, estimations for the required more detailed quantities are derived using the general method of Ref.~\cite{MaLo2005}. Thus, in comparison with the case of no mismatch, the general method of Ref.~\cite{MaLo2005} still works because it does not rely on any assumptions on the efficiency of detectors, but if we have another analytic formula for the secret key rate for the case of single photon pulses, then we need to derive  decoy state estimations for all additional quantities entering this formula.

Formula~(\ref{EqMain}) for the case of single-photon Alice's pulses also requires some observables (namely, $p_1$, $q$, and $p_{01}$) which enters neither the corresponding formula~(\ref{EqIdeal}) for the case of no mismatch nor the expression for the single-photon Alice's pulses from Ref.~\cite{Bochkov}. So, we need to derive the bounds for these observables using the decoy state method. This is the only difference of the decoy state method in the case of detection-efficiency mismatch.

Redenote the quantities $p_{\rm det}$, $p_1$, $q$, and $p_{01}$ by ${}_1p_{\rm det}$, ${}_1p_1$, ${}_1q$, and ${}_1p_{01}$ where the left subindex~1 denotes that these quantities are conditioned on the single-photon Alice's pulse (exactly as in Eqs.~(\ref{EqGamma1})--(\ref{EqGamma4})). Denote ${}^{v}p_{\rm det}$, ${}^{v}p_1$, ${}^{v}q$, and ${}^vp_{01}$ the analogous quantities conditioned on the event that Alice sends a weak coherent pulse of the type $v\in\{\rm s,d_1,d_2\}$ (i.e., with either the signal intensity $\mu$ or one of two decoy intensities $\nu_1$ and $\nu_2$). Finally, denote 
\begin{gather}
{}_1^{\rm s}p_{\rm det}=\mu e^{-\mu}\times{}_{1}p_{\rm det},
\:\:\:
{}_1^{\rm s}p_{1}=\mu e^{-\mu}\times{}_{1}p_{1},\\
{}_1^{\rm s}q=\mu e^{-\mu}\times{}_{1}q,
\quad
{}_1^{\rm s}p_{01}=\mu e^{-\mu}\times{}_{1}p_{01}.
\end{gather}
These quantities has the following meanings:
${}_1^{\rm s}p_{\rm det}$ is the joint probability that a signal pulse contains a single photon and at least one of the Bob's detectors clicks (conditioned on the measurement in the basis $z$); ${}_1^{\rm s}p_1$ is the joint probability that a signal pulse contains a single photon and Bob obtains a single click of detector one (conditioned on the measurement in the basis $z$), etc.

Let us adapt the formula (\ref{EqMain}) for the secret key rate in the case of single-photon pulses to the case of a coherent light source. Formula (\ref{EqMain}) uses the quantity of detection rate $p_{\rm det}$. However, now we should distinguish between the detection rate ${}^sp_{\rm det}$ for all signal pulses and the detection rate ${}_1^sp_{\rm det}$ for the single-photon signal pulses. 

The last term in Eq.~(\ref{EqMain}) expresses the information leak due to error correction. Error correction works with all pulses, hence, the detection rate ${}^sp_{\rm det}$ for all signal pulses should be used. Now, $Q_z$ denotes the error rate in the $z$ basis in the set of all registered signal pulses.

The first term in Ineq.~(\ref{EqMain}) expresses Eve's ignorance on the sifted key. As usual in the decoy state method, we treat multiphoton Alice's pulses as insecure. Only the signal pulses are used for key generation. Hence, all quantities in the first term in Ineq.~(\ref{EqMain}) should refer only to the single-photon signal states. 

Recall that $p_{\rm det}^{(1),\rm L}$, $\delta_x^{\rm L}$, and $p_{\rm det}^{(2),\rm U}$ (the upper limit of the range of $p_{\rm det}^{(2)}$ in the minimization) are functions of $p_{\rm det}$, $p_1$, $q$, and $p_{01}$ [see Eqs.~(\ref{EqGamma1})--(\ref{EqGamma4})], which are directly observable if the source emits true single photons. For $p_{\rm det}^{(1),\rm L}$, see Eq.~(\ref{Eqpdet1l}) and, consequently, Eq.~(\ref{Eqpdet3u}). For $\delta_x^{\rm L}$, see Eq.~(\ref{EqDeltaxl}) and, consequently, Eqs.~(\ref{Eqt1l}), (\ref{Eqp11L}), (\ref{Eqp12U}), (\ref{Eqq1u}), (\ref{Eqq2L}), and, again, Eqs.~(\ref{Eqpdet1l}) and~(\ref{Eqpdet3u}). By definition, $p_{\rm det}^{(2),\rm U}$ depend on $p_{\rm det}^{(1),\rm L}$ and $\delta_x^{\rm L}$. Thus, ultimately, $p_{\rm det}^{(1),\rm L}$, $\delta_x^{\rm L}$, and $p_{\rm det}^{(2),\rm U}$ are functions of the four mentioned arguments.

Now we should replace these arguments by the corresponding single-photon quantities (which are not directly observable anymore in the case of a coherent light source): ${}_1^{\rm s}p_{\rm det}$, ${}_1^{\rm s}p_1$, ${}_1^{\rm s}q$, and ${}^{\rm s}_1p_{01}$. The functions $p_{\rm det}^{(1),\rm L}$, $\delta_x^{\rm L}$, and $p_{\rm det}^{(2),\rm U}$ remain the same, but now they depend on these new four arguments. Thus, the secret key rate is then lower bounded by
\begin{equation}\label{EqDecoyPre}
K\geq \min_{p_{\rm det}^{(2)}
\in\big[0,p_{\rm det}^{(2),\rm U}\big]
} 
p^{(1),\rm L}_{\rm det}
\left[
1-h\left(\frac{1-\delta_x^{\rm L}}2\right)\right]
-{}^sp_{\rm det}h(Q_z),
\end{equation}
where $p_{\rm det}$, $p_1$, $q$, and $p_{01}$ should be substituted by ${}_1^{\rm s}p_{\rm det}$, ${}_1^{\rm s}p_1$, ${}_1^{\rm s}q$, and ${}^{\rm s}_1p_{01}$ in all quantities (functions of these four arguments) (\ref{Eqpdet3u}), (\ref{Eqq2L})--(\ref{Eqq1u}), (\ref{Eqp11L}), and (\ref{Eqt1l}). Thus, the first term in Ineq.~(\ref{EqDecoyPre}) looks like the same as that of  Ineq.~(\ref{EqMain}), but, in fact, all arguments are now refer only to the single-photon signal states.

The directly observable quantities are now ${}^vp_{\rm det}$, ${}^vp_1$, ${}^vq$, and ${}^vp_{01}$, $v\in\{\rm s,d_1,d_2\}$. Since the quantities ${}_1^{\rm s}p_{\rm det}$, ${}_1^{\rm s}p_1$, ${}_1^{\rm s}q$, and ${}_1^{\rm s}p_{01}$ are  not directly observable, they should be estimated. For each quantity, we should decide whether we require bounds from above, from below or both depending on the monotonicity properties of the right-hand side of Eq.~(\ref{EqDecoyPre}) with respect to these quantities.

The right-hand side of  Ineq.~(\ref{EqDecoyPre}) is an increasing function of $\delta_x^{\rm L}$ because $h(x)$ is an increasing function whenever $x\in[0,1/2]$. The condition $\delta_x^{\rm L}>0$, which ensures that the argument of $h$ falls into this segment, is a condition of Theorem~\ref{ThMain}. Otherwise, Theorem~\ref{ThMain} cannot guarantee a positive secret rate, see a comment after the theorem. Hence, the right-hand side of  Ineq.~(\ref{EqDecoyPre})  is an increasing function of ${}_1^{\rm s}p_1$ and a decreasing function of ${}_1^{\rm s}p_{01}$ and ${}_1^{\rm s}q$. So, we should estimate ${}_1^{\rm s}p_1$ from below and ${}_1^{\rm s}q$ and ${}_1^{\rm s}p_{01}$ from above. 

The dependence of the right-hand side of  Ineq.~(\ref{EqDecoyPre}) on ${}_1^sp_{\rm det}^{(1)}$ is not obvious. So, we should estimate ${}_1^sp_{\rm det}^{(1)}$ both from below and from above. A lower estimate ${}_1^{\rm s}p_{\rm det}^{\rm L}$ is derived below, while, for an upper estimate, we can take a trivial one ${}^{\rm s}p_{\rm det}$ (single-photon signal detections is a subset of all signal detections). Then, we should minimize also over ${}_1^{\rm s}p_{\rm det}\in[{}_1^{\rm s}p_{\rm det}^{\rm L},{}^{\rm s}p_{\rm det}]$. However, in our simulation given below, the minimum over ${}_1^{\rm s}p_{\rm det}$ is always achieved in the lower bound ${}_1^{\rm s}p_{\rm det}^{\rm L}$.

Let us discuss each quantity to be estimated.

(1)~~${}_1p_{\rm det}$ is the usual yield of single-photon states (denoted by $Y_1$ in Refs.~\cite{MaLo2005} and~\cite{Ma2017}), ${}^vp_{\rm det}$ is the overall gain of the pulses of a given type (denoted by $Q_\mu$, $Q_{\nu_1}$ and $Q_{\nu_2}$ in Ref.~\cite{MaLo2005} and by $Q^{a}$ in Ref.~\cite{Ma2017}), and ${}_1^{\rm s}p_{\rm det}$ is the single-photon contribution to ${}^{\rm s}p_{\rm det}$ (denoted by $Q_1^\mu$ in Ref.~\cite{MaLo2005} and by $Q_1^{s}$ in Ref.~\cite{Ma2017}). The only difference is that all our quantities are defined conditioned on the Bob's choice of the $z$ measurement basis. As we noted above, derivations in the decoy state method are still valid in the case of detection-efficiency mismatch because they do not rely on the assumption of equal efficiencies at all.  Thus, we can use the bound derived in Ref.~\cite{MaLo2005}. Conditioning on the choice of the $z$ basis for measurement also does not affect the derivation since all relations used in the derivations in Ref.~\cite{MaLo2005} are still true if we fix a basis. So,
\begin{multline}\label{EqDecoypdet}
{}_1^{\rm s}p_{\rm det}\geq 
{}_1^{\rm s}p_{\rm det}^{\rm L}=
\frac{\mu^2e^{-\mu}}{\mu\nu_1-\mu\nu_2-\nu_1^2+\nu_2^2}
\\
\times
\left[
{}^{\rm d_1}p_{\rm det}e^{\nu_1}
-{}^{\rm d_2}p_{\rm det}e^{\nu_2}
-\frac{\nu_1^2-\nu_2^2}{\mu^2}
({}_1^{\rm s}p_{\rm det}e^\mu-Y_0^{\rm L})
\right],
\end{multline}
where
\begin{equation}\label{EqDecoyY0}
Y_0^{\rm L}=\max\left[
\frac{{}^{\rm d_2}p_{\rm det}\nu_1e^{\nu_2}
-{}^{\rm d_1}p_{\rm det}\nu_2e^{\nu_1}}{\nu_1-\nu_2},
0
\right].
\end{equation}

(2)~~Derivation of the lower bound for ${}_1^{\rm s}p_{1}$ is completely the same. The quantity ${}_1^{\rm s}p_{\rm det}$ is the single-photon contribution to the gain of the signal states conditioned on the Bob's choice of the $z$ basis. Analogously, the quantity ${}_1^{\rm s}p_1$ is the single-photon contribution to the gain of the signal states detected solely by detector~1 conditioned on the Bob's choice of the $z$ basis.  Conditioning on the single clicks of detector~1 also does not affect the derivation in Ref.~\cite{MaLo2005}. Hence, the lower bound for  ${}_1^{\rm s}p_{1}$ (denoted as ${}_1^{\rm s}p_{1}^{\rm L}$) is also given by formulas~(\ref{EqDecoypdet}) and~(\ref{EqDecoyY0}) where all ${}^vp_{\rm det}$ are substituted by ${}^vp_1$.

(3)~~An upper bound for ${}_1^{\rm s}q$ was derived in Ref.~\cite{Bochkov} also as a simple generalization of derivations in Ref.~\cite{MaLo2005}:
\begin{multline}\label{EqDecoyq}
{}_1^{\rm s}q\leq {}_1^{\rm s}q^{\rm U}=
\left[\left(
{}^{\rm d_1}q^{(0)}+{}^{\rm d_1}q^{(1)}/\eta
\right)e^{\nu_1}\right.
\\\left.-
\left(
{}^{\rm d_2}q^{(0)}+{}^{\rm d_2}q^{(1)}/\eta
\right)
e^{\nu_2}
\right]
\frac{\mu e^{-\mu}}{\nu_1-\nu_2},
\end{multline}
where ${}^vq^{(\beta)}$, $\beta\in\{0,1\}$, is the joint probability that (i) Alice sends the pulse encoding bit $1-\beta$ and (ii) Bob obtains a click with the erroneous result $\beta$ (i.e., he obtains either a single click of detector $\beta$, or a double click with $\beta$ as the result of the random bit assignment), conditioned on the event that the pulse is of the type $v$ and the choice of the $x$ basis by both legitimate parties. 

(4)~~As for ${}_1^{\rm s}p_{01}$, we can simply estimate it from above as ${}_1^{\rm s}p_{01}\leq {}^{\rm s}p_{01}$. This estimation has the following intuitive interpretation. We treat multiphoton Alice's pulses as insecure. In the analysis of the case of single-photon Alice's pulses, we treated positions with multiphoton Bob's inputs as insecure. So, the most pessimistic assumption is that all multiphoton Bob's inputs originate from single-photon Alice's pulses: This assumption maximizes the number of insecure positions. 

Thus, we have obtained all required estimations and finished the adaptation of the decoy state method to the case of detection-efficiency mismatch.

\subsection{Simulation}
 
The results of calculations of the secret key rate for the decoy state protocol is given on Fig.~\ref{Fig3}. The parameters have been chosen as follows: the intensity of the signal state $\mu=0.5$, the intensities of two decoy states  $\nu_1=0.1$ and $\nu_2=0$, the fiber attenuation coefficient $\delta=0.2$~dB/km, additional losses in the Bob's optical scheme $\delta_{\rm Bob}=5$~dB, the efficiencies of the detectors $\eta_0=0.1$ and $\eta_1=0.09$ (i.e., $\eta=\eta_1/\eta_0=0.9$),  and the dark count probability per pulse for each detector $Y_0^{(0)}=Y_0^{(1)}=10^{-6}$. For simplicity of the simulation, we neglect the optical error probability, i.e., assume that the interferometer is adjusted perfectly. This probability is typically small, so the inclusion of it will lead to corrections of a higher order of smallness. The choice of the value 0.9 for the mismatch parameter $\eta$ is based on the simulations of \cite{Leuchs} and experimental results of \cite{LoHack} as well as on private communication with experimentalists: Since we do not consider the mismatch induced by Eve, the mismatch due to imperfect manufacturing and setup is typically small (the setup can be calibrated well) and $\eta\geq0.9$.

For the calculation of the actual values  of these quantities, we employ the standard model of losses and errors in a fiber-based QKD setup; see, e.g., Ref.~\cite{MaLo2005}. The probability that a photon emitted by Alice will reach the Bob's detectors is $10^{-(\delta l+\delta_{\rm Bob})/10}$, where $l$ is the transmission distance in kilometers.  Since this probability   is rather small even for $l=0$ and is very small for realistic distances, the probability that an Alice's $i$-photon state  reaches the  Bob's detectors can be approximately taken as $n10^{-(\delta l+\delta_{\rm Bob})/10}$. This approximation actually means that we neglect the possibility that more than one photon from  the Alice's pulse will reach the  Bob's detectors. The quantity $n10^{-(\delta l+\delta_{\rm Bob})/10}$ should be multiplied by the detector efficiency $\eta_\beta$ to obtain the detection probability provided that Alice's pulse contains exactly $n$ photons.
 
Then the actual values ${}^vp_\beta$, $\beta\in\{0,1\}$, are given by
\begin{eqnarray}
{}^vp_\beta
&=&Y^{(\beta)}_0+
\frac12\sum_{n=0}^\infty\left(\frac{\mu_v^n}{n!}e^{-\mu_v}\right)n
10^{-(\delta l+\delta_{\rm Bob})/10}
\eta_{\beta}
\nonumber\\
&=&Y^{(\beta)}_0+
\frac12\mu_v10^{-(\delta l+\delta_{\rm Bob})/10}
\eta_{\beta}\label{Eqp1vActual}
\end{eqnarray} 
Here, the first term is the probability of a dark count in the corresponding detector. The second term is a probability of a registration of a photon. The summation means averaging over the number of photons $n$ according to the Poisson distribution, while the term for a given $n$ was explained above. In the second line, we substitute the sum  by a well-known expectation for the Poisson distribution: $\sum_{n=0}^\infty\frac{\mu_v^n}{n!}e^{-\mu_v}n=\mu_v$. The factor $1/2$ in the second term of the right-hand side of Eq.~(\ref{Eqp1vActual}) is the probability that Alice's bit is equal to $\beta$ (hence, the pulse reaches the detector $\beta$ and not the detector $1-\beta$).

The right-hand side of Eq.~(\ref{Eqp1vActual}) expresses the fact the the detector clicks either due to a dark count or due to a detection of a photon. We neglect the probability of the joint events of a dark count and a detection of a photon since both probabilities are small and this joint probability is of the second order of smallness. We stress that this assumption is not a part of the security proof and used only for the simulation.

Then we have ${}^vp_{\rm det}={}^vp_0+{}^vp_1$. Here the probability of a double click has a higher order of smallness and, hence, is neglected.

However, we should estimate the mean probability of a double click for the formula (\ref{EqDecoyPre}) because, in this formula, it is required \textit{per se}, not as a small correction to another term. It occurs in the case of either two dark counts, or detection of two photons by different detectors, or a dark count in one detector and a detection of a photon by the other detector. The dark count probability and the probability of detection of a photon are approximately of the same order. All terms below are of the second order of smallness. We neglect the terms of the higher orders of smallness (i.e., corresponding to a detection of more than two photons, a detection of two photons and a simultaneous dark count, etc.). Recall that, according to Eq.~(\ref{EqGamma4}), ${}^sp_{01}$ is the mean probability of a double click for the two bases (and, in our case, for a pulse of the signal intensity). Denote ${}^sp^{(b)}_{01}$ denotes the double click probability provided that Bob measures in the basis $b\in\{z,x\}$. Then
\begin{equation}
{}^sp_{01}=\frac{{}^sp^{(z)}_{01}+{}^sp^{(x)}_{01}}2,
\end{equation}
\begin{equation}\label{EqDecoypdcsim}
\begin{split}
&{}^sp^{(b)}_{01}=Y_0^{(0)}Y_0^{(1)}
\\
&+
\frac{Y_0^{(0)}\eta_1+Y_0^{(1)}\eta_0}2
\mu
10^{-(\delta l+\delta_{\rm Bob})/10}
+
\\
&+\frac12\sum_{n=2}^\infty \frac{n(n-1)}2
\frac{\mu^n}{n!}e^{-\mu}
\left(
10^{-(\delta l+\delta_{\rm Bob})/10}
\right)^2
\\
&\times
(1-p_b)\eta_0\eta_1
\\
&=Y_0^{(0)}Y_0^{(1)}+
\frac{Y_0^{(0)}\eta_1+Y_0^{(1)}\eta_0}2
\mu
10^{-(\delta l+\delta_{\rm Bob})/10}
+
\\
&+\frac{\mu^2}4
10^{-2(\delta l+\delta_{\rm Bob})/10}
(1-p_b)\eta_0\eta_1.
\end{split}
\end{equation}
where, as before, $p_b$, $b\in\{z,x\}$, is the probability of choosing the basis $b$. Note that ${}^sp_{01}$ does not depend on $p_b$.

In Eq.~(\ref{EqDecoypdcsim}), the first term $Y_0^{(0)}Y_0^{(1)}$ is the probability of the double clicks in both detectors. The second term is the probability of a photon registration in one detector and a double click in the other detector averaged over the number of photos in the pulse, see Eq.~(\ref{Eqp1vActual}). The factor $1/2$ in the second term is again [like in Eq.~(\ref{Eqp1vActual})] the probability of the coincidence of Alice's bits and Bob's detector. More precisely, consider, for definiteness, the event of a dark count in detector~0 and a photon registration in detector~1, which corresponds to the term proportional to $Y_0^{(0)}\eta_1$. If the bases of Alice and Bob coincide, this event may occur only if Alice sent bit~1, which has the probability $1/2$. If the bases of Alice and Bob do not coincide, then, independently on Alice's bit value, her photon arrives at detector~1 with probability $1/2$. Thus, in any case, we obtain the factor $1/2$.

The third term in Eq.~(\ref{EqDecoypdcsim}) is the probability of detection of two photons by different detectors, also averaged over the number of photos in the pulse. Since we neglect the optical error, a detection of two photons by different detectors may occur only if the Alice's and Bob's bases are different. If Bob chooses basis $\beta$, the factor $1-p_\beta$, $\beta\in\{z,x\}$, is a probability that Alice chooses the other basis. In this case, the probability that two photons, which have reached the detectors, ``choose'' different detectors, is $1/2$, which explains this prefactor before this term. If the pulse has $n$ photons, then the binomial coefficient $n(n-1)/2$ is a number of ways to choose two photons from $n$. The factor $10^{-2(\delta l+\delta_{\rm Bob})/10}$ is the probability that a given pair of photons reaches Bob's detectors. We have used the well-known expression for the second factorial moment of the Poisson distribution: $\sum_{n=2}^\infty n(n-1)\frac{\mu^n}{n!}e^{-\mu_v}=\mu_v^2$.

Since we neglect the optical error, an erroneous detection occurs only in the case of a dark count (which is erroneous with the probability $1/2$):
${}^vq^{(\beta)}=Y_0^{(\beta)}/2$.

On Fig.~\ref{Fig3}, we compare the secret key rate according to  formula (\ref{EqDecoyPre}) (with the decoy state method estimates) and the secret key rate in the case of no efficiency mismatch but the same average detection efficiency $(\eta_0+\eta_1)/2$ (i.e., formula (\ref{EqIdeal}) combined with the usual decoy state estimates). The ratio of the secret key rate in the first case to that in the second case (an analogue of Fig.~\ref{Fig2}) is shown on Fig.~\ref{Fig4}. Again (as on Figs.~\ref{Fig1} and~\ref{Fig2}), from Figs.~\ref{Fig3} and~\ref{Fig4}, we see that the reduction of the secret key rate due to detection-efficiency mismatch is almost negligible (but still strictly positive and, hence, should be estimated by the presented methods) whenever the mismatch is small.

The initial increase of the ratio on Fig.~\ref{Fig4} (with the maximum on the distance approximately 80~km) is caused by the decrease of the double click rate for large distances due to transmission loss. Let us explain this. Some Alice's pulses are multiphoton. For large distances, the probability that two photons achieve Bob's lab is negligible. However, for small distances, it may be significant. Multiphoton pulses arriving at Bob's lab cause double clicks [corresponds to the third term in Eq.~(\ref{EqDecoypdcsim})]. Double clicks on Bob's side are treated as insecure, so, the estimated secret key rate decreases if the double click rate increases. 

One can say that, in our formalism, positions where Alice emits a multiphoton pulse and at least two  photons from it achieve Bob's lab  reduce the estimated secret key rate twice. Firstly, Alice's multiphoton pulses are treated as insecure in the decoy state method. Secondly, if two or more photons in a pulse achieves Bob's lab (the number of such pulses is estimated mainly by the double click rate), then, as we explained in Sec.~\ref{SecDecoyEst}, we pessimistically treat them as originated from the single-photon Alice's pulses. In other words, we pessimistically assume that Alice sent a single-photon pulse and additional photons were added by Eve and treat this pulse as insecure. This additionally reduces the estimated secret key rate for the single-photon Alice's pulses.

\begin{figure}[t]
\begin{centering}
\includegraphics[width=1\columnwidth]{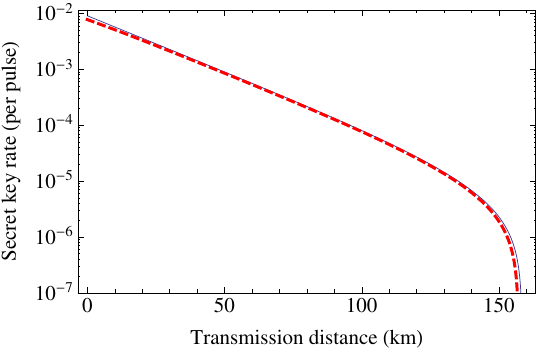}
\end{centering}
\vskip -4mm
\caption
{
Secret key rate of the decoy state BB84 protocol with detection-efficiency mismatch. The parameters are as follows: the intensity of the signal state, $\mu=0.5$; the intensities of two decoy states,  $\nu_1=0.1$ and $\nu_2=0$; the fiber attenuation coefficient, 0.2~dB/km; additional losses in the Bob's optical scheme, 5~dB; the efficiencies of the detectors, $\eta_0=0.1$ and $\eta_1=0.09$ (i.e., $\eta=\eta_1/\eta_0=0.9$); and  the dark count probability per pulse for each detector, $Y_0^{\beta=0}=Y_0^{\beta=1}=10^{-6}$. Red dashed line: formula (\ref{EqDecoyPre}) with decoy method estimates. Blue line: the case of no efficiency mismatch but the same average detection efficiency $(\eta_0+\eta_1)/2$ (like on Fig.~\ref{Fig2}).
}
\label{Fig3}
\end{figure}

\begin{figure}[t]
\begin{centering}
\includegraphics[width=1\columnwidth]{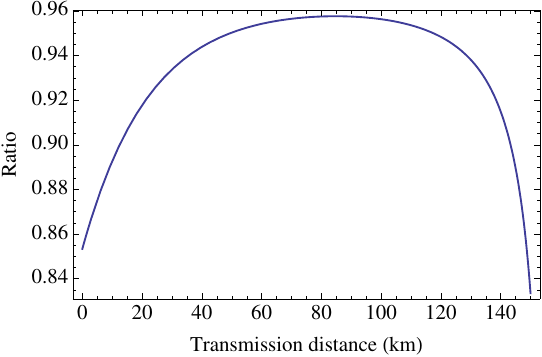}
\end{centering}
\vskip -4mm
\caption
{
Decrease of secret key rate in the detection efficiency-mismatch case with respect to the no-mismatch case in the case of decoy state protocol: the ratio of the secret key rate in the mismatch case with the detector efficiencies 1 and $\eta$ to the secret key rate in the no-mismatch case with both  efficiencies equal to $(1+\eta)/2$ (like on Fig.~\ref{Fig2}). The parameters are as on Fig.~\ref{Fig3}.
}
\label{Fig4}
\end{figure}

\section{Conclusions}

We have solved an important problem for practical QKD and rigorously proved the security of the BB84 protocol with detection-efficiency mismatch for the multiphoton case on both Alice's output and Bob's input. The main formula of the present paper is (\ref{EqMain}). It is formulated for the case of single-photon Alice's pulses. The second important development is an adaptation of the decoy state method to the case of detection-efficiency mismatch. 

We have shown that the reduction of the secret key rate due to detection-efficiency mismatch is almost negligible whenever the mismatch is small and the QBER is far from the known critical value 11\%.

The proposed new methods are not restricted to this particular problem and can be used in other QKD security proofs. The security proof for a QKD protocol is reduced to a convex optimization problem. Proposition~\ref{PropSym} (with Remark~\ref{RemSym}) allows one to reduce the dimensionality of the search space and, thus, can be used in both analytic proof and numerical approach to QKD \cite{Lutk-num-pre,Lutk-num,EntanglVerif,Zhang}. Propositions~\ref{Propdb} and~\ref{PropQ2} allows one to rigorously estimate the number of multiphoton events in various contexts. A rigorous estimation of the number of multiphoton events was the lacking part of the analysis of Ref.~\cite{Zhang}. Multiphoton attacks in mistrustful quantum cryptography were considered also in \cite{DiamantiMultiphot}. It would be interesting to find applications of Propositions~\ref{Propdb} and~\ref{PropQ2} in mistrustful quantum cryptography as well.

An open problem is to include the dark count rates and, in particular, dark count rate mismatch in the detection model. The difference between dark count rates of two detectors also may affect the security and the secret key rate.

Another open problem is generalization to the case when the detection-efficiency mismatch is not constant, but is under partial Eve's control, or, in other words, Eve-induced mismatch rather than mismatch only due to manufacturing and setup. Such attacks are described and employed experimentally \cite{Leuchs,LoHack,Makarov,Pirandola-rev,Xu}. Ref.~\cite{Zhang} analyzes the security in this general case under the same conjecture about the number of multiphoton events supported by numerical evidence. That is, the idea that the number of multiphoton events can be estimated using the number of double clicks and/or the error rate averaged over the two bases works also in the general case. One might hope that Proposition~\ref{Propdb} or its modification can be used to obtain a rigorous bound for the number of multiphoton events in this case as well. 

Finally, let us note that, originally, the Devetak--Winter formula for the secret key rate (\ref{EqDW}) is valid only for collective Eve's attacks: Eve prepares the state $\rho^{\otimes N}_{ABE}$ where $N$ is the number of sendings. For the asymptotic case $N\to\infty$, the most general, coherent attacks are reduced to the collective attacks using the quantum version of the de Finetti representation \cite{Renner}, so, this is not an actual restriction. For finite $N$, the entropy accumulation technique \cite{EntrAccum,EntrAccum2,EntrAccum3} can be used for better finite-size corrections. However, a concrete formula using this technique is to be elaborated. 
 
%\section*{Acknowledgments}
\begin{acknowledgments}
The author is grateful to Norbert L\"utkenhaus and Yanbao Zhang for fruitful discussions. This work was funded by the Ministry of Science and Higher Education of the Russian Federation (grant number 075-15-2020-788).
\end{acknowledgments}
\bigskip
\appendix

\section{Proof of Lemma~\ref{Lemdbmonot}}\label{AppLemdbmonot}

Let us define the function
\begin{equation}
F(x,y)=2y\log(2^{x-1}-1)+2h(y)-x+2
\end{equation}
for $x>0$ and $y\in[0,1]$. We can see that
$F(x,0)<0$ and $F(x,1/2)>0$ for $x\geq3$. Hence, for any fixed $x\geq3$, $F(x,y)$ has a  root denoted by $y=f(x)$. Since $F(x,1)>0$ for $x\geq3$ and $F$ is concave with respect to $y$, there are no more roots of $F(x,y)$. For $x=n$ denote $f(x)=p^{{\rm min}, (n)}_{01}$. 

Now let us prove that $p^{{\rm min}, (n)}_{01}$ is a non-decreasing function of $n$ for $n\geq3$. It is sufficient to prove that $f'(x)\geq0$ for $x\geq3$. We have
\begin{equation}
f'(x)=-\frac{F'_x(x,f(x))}{F'_y(x,f(x))}.
\end{equation}
$F'_y(x,y)>0$ for $y\leq1/2$,
\begin{equation}
F'_x(x,y)=\frac{2y}{1-2^{-(x-1)}}-1<0
\end{equation}
for $y<\frac12[1-2^{-(x-1)}]$.
Hence, we should prove that 
\begin{equation}\label{EqLemImplicit}
f(x)<\frac{1-2^{-(x-1)}}2
\end{equation}
for $x\geq 3$. Indeed, since $h(y)>0$ for $0<y<1$ and, as can be proved,
\begin{equation}
(1-2^{-(x-1)})\log(2^{x-1}-1)-x+2\geq0
\end{equation}
with the equality only in the point $x=2$,
we obtain $F(x,(1-2^{-(x-1)})/2)>0$ for $x\geq3$, which implies Ineq.~(\ref{EqLemImplicit}). The lemma has been proved.

\section{Proof of Lemma~\ref{Lempdet3}}\label{Apppdet3}

Since we have only two free parameters $p_{\det}^{(2)}$ and $p_{\det}^{(3+)}$, we treat $\tilde\delta_x^{\rm L}$ as a function of these two parameters.

We need to prove two things: 
\begin{enumerate}[(i)]
\item Well-definiteness of the expression under minimization in Ineq.~(\ref{EqMain}) for all $p_{\rm det}^{(2)}\in[0,p_{\rm det}^{(2),\rm U}]$.

\item The minimum of the right-hand side of Ineq.~(\ref{EqMainPre}) on $(p_{\rm det}^{(2)},p_{\rm det}^{(3+)})\in D$ is achieved on $(p_{\rm det}^{(2)},p_{\rm det}^{(3+),\rm U})$, $p_{\rm det}^{(2)}\in[0,p_{\rm det}^{(2),\rm U}]$. 
\end{enumerate}

In paragraphs 1--4 below, we prove (i).  This means that the denominator $p_{\rm det}^{(1),\rm L}$ in expression (\ref{EqDeltaxl}) for $\delta_x^{\rm L}$ is strictly positive (non-zero) and $\delta_x^{\rm L}\leq1$ on this segment. Recall that $p_{\rm det}^{(2),\rm U}$ is defined as the maximal value of $p_{\rm det}^{(2)}$ such that $p^{(1),\rm L}_{\rm det}\geq0$ and $\delta_x^{\rm L}\leq1$. So, we need to prove that $p^{(1),\rm L}_{\rm det}>0$ if $\delta_x^{\rm L}\leq1$.

In paragraphs 5--7, we prove (ii).

1. First, let us observe that $\tilde\delta_x^{\rm L}$ is a decreasing function of $p_{\rm det}^{(3+)}$. This can be proved by direct inspection of the function  $\tilde\delta_x^{\rm L}$ (note that it is rational with respect to the argument $p_{\det}^{(3+)}$) using the inequalities
\begin{equation}
p_1-p_1^{(2),\rm U}\leq
p_1-p_1^{(2)}\leq
p_{\rm det}-p_{\rm det}^{(2)}.
\end{equation}

2. Let us prove that $\tilde\delta_x^{\rm L}\leq1$ or, equivalently, 
\begin{equation}\label{Eqtqp}
\tilde t_1^{\rm L}-2q_1^{\rm U}-\frac{p_{\rm det}^{(1),\rm L}}{\sqrt\eta}\leq0
\end{equation}
in the points $\big(0,p_{\rm det}^{(3+)}\big)$ (i.e., $p_{\rm det}^{(2)}=0$ and $p_{\rm det}^{(3+)}$ is arbitrary).
Obviously, Ineq.~(\ref{Eqtqp}) is true for the point $(0,0)$ (i.e., $p_{\rm det}^{(2)}=p_{\rm det}^{(3+)}=0$) because this is the case of the single-photon Bob's input and $\tilde\delta_x^{\rm L}$ coincides with $\delta_x$ (see Eqs.~(\ref{EqDelta})). $\delta_x\leq1$ due to positivity of $\rho_{AB}^{(1)}$. Since $\tilde\delta_x^{\rm L}$ decreases with $p_{\rm det}^{(3+)}$, Ineq.~(\ref{Eqtqp}) is also satisfied in the points $\big(0,p_{\rm det}^{(3+)}\big)$.

Note that $(0,0)\in D$. Indeed, condition (\ref{Eqpdet1}) is satisfied simply because $p_{\rm det}\geq0$, condition (\ref{Eqpdet3u}) is also obviously satisfied and, as said above, $\tilde\delta_x^{\rm L}=\delta_x\leq1$.

3. Since the left-hand side of Ineq.~(\ref{Eqtqp}) is a convex function of $p_{\rm det}^{(2)}$ and due to the result of the previous paragraph, the equality in Ineq.~(\ref{Eqtqp}) can be achieved in at most one point $p_{\rm det}^{(2)}$, for each $p_{\rm det}^{(3+)}$.

4. By the definition of $p_{\rm det}^{(2),\rm U}$, either $\delta_x^{\rm L}=1$ or $p_{\rm det}^{(1),\rm L}=0$ is true for $p_{\rm det}^{(2)}=p_{\rm det}^{(2),\rm U}$. In the latter case, since, by the conditions of Theorem~\ref{ThMain}, $t_1^{\rm L}-2q_1^{\rm U}>0$ for all $p_{\rm det}^{(2)}\leq p_{\rm det}- p_{\rm det}^{(2),\rm U}$, we have $\delta_x^{\rm L}\to+\infty$ as $p_{\rm det}^{(2)}\to p_{\rm det}^{(2),\rm U}-0$. Hence, $\delta_x^{\rm L}=1$ for some $p_{\rm det}^{(2)}<p_{\rm det}^{(2),\rm U}$, which contradicts the definition of $p_{\rm det}^{(2),\rm U}$. Hence, the former alternative takes place: $\delta_x^{\rm L}=1$ for $p_{\rm det}^{(2)}=p_{\rm det}^{(2),\rm U}$ and $p_{\rm det}^{(1),\rm L}>0$ for $p_{\rm det}^{(2)}\leq p_{\rm det}^{(2),\rm U}$.

From paragraphs~2, 3 and~4, it follows that $p_{\rm det}^{(1),\rm L}>0$ and $\delta_x^{\rm L}\leq1$ (and, hence, the expression under minimization in Ineq.~(\ref{EqMain}) is well-defined) for all $p_{\rm det}^{(2)}\leq p_{\rm det}^{(2),\rm U}$. 

5.  Let us prove that 
\begin{equation}\label{Eqpdet2uineq}
\big(p_{\rm det}^{(2)},p_{\rm det}^{(3+)}\big)\notin D  \text{ whenever } p_{\rm det}^{(2)}>p_{\rm det}^{(2),\rm U}.
\end{equation}
By the results of the paragraph~4, $\tilde\delta_x^{\rm L}>1$ for $\big(p_{\rm det}^{(2)},p_{\rm det}^{(3+),\rm U}\big)$ whenever $p_{\rm det}^{(2)}>p_{\rm det}^{(2),\rm U}$. Since $\tilde\delta_x^{\rm L}$ is a decreasing function of $p_{\rm det}^{(3+)}$, the same is true for all pairs $\big(p_{\rm det}^{(2)},p_{\rm det}^{(3+)}\big)$ with $p_{\rm det}^{(2)}>p_{\rm det}^{(2),\rm U}$ and $p_{\rm det}^{(3+)}\leq p_{\rm det}^{(3+),\rm U}$, which proves Eq.~(\ref{Eqpdet2uineq}).

That is, $p_{\rm det}^{(2)}\leq p_{\rm det}^{(2),\rm U}$ is a restriction for an arbitrary $p_{\rm det}^{(3+)}$, not only  for $p_{\rm det}^{(3+)}=p_{\rm det}^{(3+),\rm U}$.

6. The conditions of Theorem~\ref{ThMain} ensure that $q_1^{\rm U}<t_1^{\rm L}/2$ for all $p_{\rm det}^{(2)}\in[0,p_{\rm det}^{(2),\rm U}]$. Since $p_{\rm det}^{(3+),\rm U}\geq p_{\rm det}^{(3+)}$, we have $\tilde p_1^{(1),\rm L}\geq p_1^{(1),\rm L}$ [compare Eqs.~(\ref{Eqp11Lderiv}) and (\ref{Eqp11L})], and, consequently, $\tilde t_1^{\rm L}\geq t_1^{\rm L}$ [compare Eqs.~(\ref{Eqt1Lderiv}) and (\ref{Eqt1l})]. Hence, $q_t^{\rm U}<\tilde t_1^{\rm L}/2$ and, thus, $\tilde\delta_x^{\rm L}>0$ in all points $\big(p_{\rm det}^{(2)},p_{\rm det}^{(3+),\rm U}\big)$ such that $p_{\rm det}^{(2)}\leq p_{\rm det}^{(2),\rm U}$. Since $\tilde\delta_x^{\rm L}$ decreases with $p_{\rm det}^{(3+)}$ and due to Eq.~(\ref{Eqpdet2uineq}), $\tilde\delta_x^{\rm L}>0$ in $D$.

Hence, the argument $(1-\tilde\delta_x^{\rm L})/2$ of the function $h$ in Ineq.~(\ref{EqMainPre}) belongs to the segment $[0,1/2]$ whenever $(p_{\rm det}^{(2)},p_{\rm det}^{(3+)})\in D$. Recall that $h$ monotonically increases on this segment. This will be used in the next  paragraph.

7. Then, the right-hand side of Ineq.~(\ref{EqMainPre}) is a decreasing function of $p_{\rm det}^{(3+)}$ because both the factor $p_{\rm det}^{(1)}$ before $(1-h)$ and $\tilde\delta_x^{\rm L}>0$ are decreasing functions of $p_{\rm det}^{(3+)}$. Hence, the maximal possible value $p_{\rm det}^{(3+),\rm U}$ of $p_{\rm det}^{(3+)}$ corresponds to the minimal secret key rate for a given $p_{\rm det}^{(2)}$.

Hence, minimization of the right-hand side of Ineq.~(\ref{EqMainPre}) over $\big(p_{\rm det}^{(2)},p_{\rm det}^{(3+)}\big)\in D$ is reduced to the minimization over $p_{\rm det}^{(2)}\in\left[0,p_{\rm det}^{(2),\rm U}\right]$ with $p_{\rm det}^{(3+)}=p_{\rm det}^{(3+),\rm U}$. This proves formula~(\ref{EqMain}).

\section{Proof of Lemma~\ref{LemConcave}}\label{AppConcave}

For arbitrary $x,y$ and an arbitrary $0\leq t\leq 1$, we have
\begin{equation}
\begin{split}
&txh\left(\frac12-\frac{g(x)}x\right)+
(1-t)yh\left(\frac12-\frac{g(y)}y\right)
\\&=
[tx+(1-t)y]
\bigg\lbrace
\frac{tx}{tx+(1-t)y}h\left(\frac12-\frac{g(x)}x\right)
\\&\hspace{2.5cm}+
\frac{(1-t)y}{tx+(1-t)y}h\left(\frac12-\frac{g(y)}y\right)
\bigg\rbrace
\\&\leq
[tx+(1-t)y]
h\left(
\frac12-
\frac{tg(x)+(1-t)g(y)}{tx+(1-t)y}
\right)
\\&\leq
[tx+(1-t)y]
h\left(
\frac12-
\frac{g(tx+(1-t)y)}{tx+(1-t)y}
\right),
\end{split}
\end{equation}
q.e.d. The first inequality comes from concavity of $h$ and the second inequality comes from convexity of $g$ and monotonicity of the function $h(x)$ on the segment $[0,1/2]$.

\end{document}